\documentclass[a4paper,USenglish,cleveref,autoref,thm-restate]{lipics-v2021}

\hideLIPIcs  

\title{Cup Product Persistence and Its Efficient Computation} 

\titlerunning{Cup Product Persistence and Its Efficient Computation} 

\author{Tamal K. Dey }{Department of Computer Science, Purdue University, USA \and }{tamaldey@purdue.edu}{}{}

\author{Abhishek Rathod}{Department of Computer Science, Ben Gurion University, Israel}{arathod@post.bgu.ac.il}{}{}

\authorrunning{T.K. Dey and A. Rathod} 

\Copyright{Jane Open Access and Joan R. Public} 


\funding{This work is partially supported by the NSF grant CCF 2049010, DMS 2301360 and ERC 101039913.}

\acknowledgements{We wish to acknowledge helpful initial discussions with Ulrich Bauer and Fabian Lenzen.} 

\nolinenumbers 

\usepackage{lineno}

\usepackage[utf8]{inputenc}
\usepackage[T1]{fontenc}

\usepackage[pdftex,dvipsnames,table]{xcolor}
\usepackage[normalem]{ulem}

\usepackage{adjustbox}

\usepackage{amsmath}

\usepackage{soul}

\definecolor{MyRed}{HTML}{BF616A} 
\definecolor{MyBlue}{HTML}{5E81AC}
\definecolor{qqttcc}{rgb}{0,0.2,0.8}
\definecolor{ttttff}{rgb}{0.2,0.2,1}
\definecolor{wrwrwr}{rgb}{0.3803921568627451,0.3803921568627451,0.3803921568627451}
\definecolor{rvwvcq}{rgb}{0.08235294117647059,0.396078431372549,0.7529411764705882}

\newcommand{\cancel}[1]

\DeclareMathOperator{\Rips}{VR}

\DeclareMathOperator{\ripsfull}{Rips}

\DeclareMathOperator{\inter}{I}

\DeclareMathOperator{\diam}{diam}

\DeclareMathOperator{\Cech}{\check{\mathsf{C}}}
\DeclareMathOperator{\cechfull}{\check{C}{ech}}

\usepackage{float}

\usepackage{todonotes}

\usepackage{tikz,tikz-cd}
\usetikzlibrary{arrows,automata}
\usetikzlibrary{matrix,backgrounds}
\tikzset{tail reversed/.code={\pgfsetarrowsstart{tikzcd to}}}

\usepackage{float}

\usepackage{paralist}

\usepackage{amsmath,amsthm,mathtools,extpfeil}

\usepackage{amssymb,mathrsfs}

\usepackage{nicefrac}

\usepackage{todonotes}

\usepackage{url}

\usepackage{sectsty}
\sectionfont{\fontsize{12}{15}\selectfont}

\usepackage[all]{xy}

\newcommand{\cochainmap}{\varphi}

\newcommand{\cohommap}{\varphi^{\ast}}

\newcommand{\cohommaptwo}{\psi^{\ast}}

\theoremstyle{remark}
\newtheorem{fact}{Fact}

\usepackage[vlined, ruled,linesnumbered]{algorithm2e}
\usepackage{algpseudocode}

\usepackage{xr}
\usepackage{hyperref}
\usepackage{cleveref}

\usepackage{babel}

\newcommand{\cupmap}{\smile_{\bullet}}

\newcommand{\cupmapi}{\smile_{i}}

\newcommand{\rpn}{\mathbb{RP}^n}

\newcommand{\cpn}{\mathbb{CP}^n}

\DeclareMathOperator*{\DP}{DP}

\DeclareMathOperator*{\op}{op}

\newcommand{\cupprod}{\mathbf{cuplength}}
\newcommand{\Int}{\mathbf{Int}}

\newcommand{\complex}{\mathsf{K}}
\newcommand{\altcomplex}{\mathsf{L}}
\newcommand{\altaltcomplex}{\mathsf{C}}
\newcommand{\rk}{\mathsf{rk}}
\newcommand{\Dgm}{\mathsf{Dgm}}

\newcommand{\Chain}{\mathsf{C}}

\newcommand{\Homol}{\mathsf{H}}

\newcommand{\Bdry}{\mathsf{B}}

\newcommand{\Cycl}{\mathsf{Z}}

\newcommand{\PCB}{\Omega}
\usepackage{cite}

\newcommand{\NBB}{\mathbb{N}}

\newcommand{\barcode}{B(\mathcal{F})}

\newcommand{\matrixh}{\mathbf{H}}

\newcommand{\matrixb}{\mathbf{S}}
\newcommand{\basisB}{S}

\newcommand{\matrixr}{\mathbf{R}}

\newcommand{\matrixc}{\mathbf{C}}

\newcommand{\productmodulek}{\im \mathsf{H}^\ast(\smile_\bullet)}

\newcommand{\productrelmodulek}{\im \rel  \mathsf{H}^\ast(\smile_\bullet)}

\newcommand{\productmodulekk}{\im \mathsf{H}^\ast(\cupmap^{k})}

\newcommand{\ordmodule}{\mathsf{H}^\ast(\complex_\bullet)}

\DeclareMathOperator{\im}{im}

\DeclareMathOperator{\rel}{rel}

\newcommand{\Simp}{\mathbf{Simp}}

\begin{document}


\maketitle

\begin{abstract} 
It is well-known that the cohomology ring has a richer structure than homology groups.
However, until recently, the use of cohomology in persistence setting has been limited to speeding up of barcode computations.
Some of the recently introduced invariants, namely, persistent cup-length~\cite{contessoto}, persistent cup modules~\cite{cuparxivtwo,memoli2} and persistent Steenrod modules~\cite{lupo2018persistence}, to some extent, fill this gap. 
When added to the standard persistence barcode, they lead to invariants that are more discriminative than the standard persistence barcode.       
In this work, we 
 devise an $O(d n^4)$ algorithm for computing the persistent $k$-cup modules for all $k \in \{2, \dots, d\}$, where $d$ denotes the dimension of the filtered complex, and $n$ denotes its size. 
Moreover, we note that since the persistent cup length can be obtained as a byproduct of our computations, this leads to a  
faster algorithm for computing it for $d>3$.  
Finally, we introduce a new stable invariant called partition modules of cup product that is more discriminative than persistent $k$-cup modules and devise an $O(c(d)n^4)$  algorithm for computing it, where $c(d)$ is subexponential in $d$.
\end{abstract}

\section{Introduction}

Persistent homology is one of the principal tools in the fast growing field of topological data analysis. A solid algebraic framework \cite{zomorodian2004computing}, a well-established theory of stability~\cite{CEH07,chazal2014persistence} along with fast algorithms and software~\cite{bauer2021ripser,bauer2014clear,boissonnat2013compressed,phat,gudhi} to compute complete invariants called barcodes of filtrations have led to the successful adoption of single parameter persistent homology as a data analysis tool~\cite{DW22,EH10}.
This standard persistence framework operates in each (co)homology degree separately
and thus cannot capture the interactions across degrees in an apparent way. 
To achieve this, one may endow a cohomology
vector space with the well-known \emph{cup product} 
 forming a graded algebra. 
 Then, the isomorphism type of such graded algebras can reveal information
including interactions across degrees. 
However, even the best known algorithms for determining isomorphism of graded algebras run in exponential time in the worst case ~\cite{brooksbank2019testing}. So it is not immediately clear how one may extract new (persistent) invariants from the product structure efficiently in practice.

Cohomology has already shown to be useful in speeding up persistence computations before\cite{bauer2014clear,bauer2021ripser,boissonnat2013compressed}. It has also been noted that additional structures on cohomology provide an avenue
to extract rich topological information\cite{yarmola2010persistence,lupo2018persistence,herscovich2018higher,belchi2021a,contessoto}.
To this end, in a recent study, the authors of \cite{contessoto}
introduced the notion of 
(the persistent version of) an invariant called the \emph{cup length}, which is the maximum number of cocycles with a nonzero product. 
In another version \cite{cuparxivtwo}, the authors of \cite{contessoto} introduced an invariant called \emph{barcodes of persistent $k$-cup modules} which are  stable, and can add more discriminating ability (Figure~\ref{fig:prodpers}).
Computing this invariant allows us to capture interactions among various degrees. 
In \Cref{ex:simp,ex:proj}, we provide simple examples for which persistent cup modules can disambiguate filtered spaces where ordinary persistence and persistent cup-length fail. Notice 
that for a filtered $d$-complex, the  $k$-cup modules for $k\in \{2,\dots,d\}$ may not be a strictly finer invariant on its own
compared to ordinary persistence. It can however add more information as~\Cref{ex:simp,ex:proj} illustrate.

\begin{example} \label{ex:simp}
    See Figure~\ref{fig:prodpers}.
    Let $\complex^1$ be a cell complex obtained by taking a wedge of four circles and two $2$-spheres. 
    Let $\complex^2$ be a cell complex obtained by taking a wedge of two circles, a sphere and a $2$-torus.
    Let $\complex^3$ be a cell complex obtained by taking a wedge of two tori. 
 \begin{remark}
    Throughout, for a cell complex $\altaltcomplex$, the filtration for which all the $k$-dimensional cells of $\altaltcomplex$ arrive at the same index is referred to as the \emph{natural cell filtration associated to $\altaltcomplex$}. 
\end{remark}   
    
    Consider the natural cell filtrations $\complex^1_\bullet$, $\complex^2_\bullet$ and $\complex^3_\bullet$.
  Standard persistence cannot tell apart $\complex^1_\bullet$, $\complex^2_\bullet$  and $\complex^3_\bullet$ as the  barcode for the three filtrations are the same. Persistent cup length cannot distinguish $\complex^2_\bullet$ from $\complex^3_\bullet$, whereas the barcodes for persistent cup modules for  $\complex^1_\bullet$, $\complex^2_\bullet$ and $\complex^3_\bullet$ are all different. 
\end{example}

\begin{example} \label{ex:proj}

We now consider another example that compares the relative discriminative strengths of standard persistence, persistent cup length and persistent cup modules.

\smallskip

\textbf{Filtered real projective space.}
The real projective space $\rpn$ is the space of lines through the origin in $\mathbb{R}^{n+1}$. It is homeomorphic to the quotient space ${S^n} / {(u \simeq -u)}$  obtained by identifying the antipodal points of a sphere, which in turn is homeomorphic to  ${D^{n}} / {(v \simeq -v)}$ for $v\in \partial D^{n}$. 
Since ${S^{n-1}} / {(u \simeq -u)} \cong \mathbb{RP}^{n-1}$,  $\rpn$ can be obtained from $\mathbb{RP}^{n-1}$ by attaching a cell $D^n$ using the projection $\wp_n: {S^{n-1}} \to \mathbb{RP}^{n-1}$.
Thus, $\rpn$ is a CW complex with one cell in every dimension from $0$ to $n$. 
This gives rise to the natural cell  filtration $\rpn_{\bullet}$ for $\rpn$, where cells of successively higher dimension  are introduced with  attaching maps $\wp_i$ for $i\in [n]$ described above. 
Finally, the cohomology algebra of $\rpn$ is given by $\mathbb{Z}_2 [x] / (x^{n+1})$, where $x \in \mathsf{H}^1(\rpn)$~\cite[pg.~146]{hausmann2014mod}.

\smallskip

\textbf{Filtered complex projective space.}
The complex projective space $\cpn$ is the space of complex lines through the origin in $\mathbb{C}^{n+1}$. It is homeomorphic to the quotient space $S^{2n+1} / S^1 \cong S^{2n+1} / (u \simeq \lambda_q u)$, which in turn can be shown to be homeomorphic to $D^{2n} / (v \simeq \lambda_q v)$ for $v \in \partial D^{2n}$ for all $\lambda_q \in \mathbb{C}$, $|\lambda_q| = 1$. Therefore, $\cpn$ is obtained from $\mathbb{CP}^{n-1}$ by attaching a $2n$-dimensional cell $D^{2n}$ using the projection $\wp'_{2n}: S^{2n-1} \to \mathbb{CP}^{n-1}$. Thus, $\cpn$ is a CW complex with one cell in every even dimension from $0$ to $2n$. This yields the natural cell filtration $\cpn_{\bullet}$ for $\cpn$ where a cell of dimension $2i$ is added to the CW complex for $i\in[n]$ with the attaching maps $\wp'_{2i}$ for $i\in[n]$ described above. The cohomology algebra of $\cpn$ is given by $\mathbb{Z}_2 [y] / (y^{n+1})$, where $y \in \mathsf{H}^2(\cpn)$~\cite[pg.~241]{hausmann2014mod}.

\smallskip

\textbf{Filtered wedge of spheres.}
Let $\altcomplex^{n} = S^{1} \vee \dots \vee S^{n}$ be a wedge of spheres of increasing dimensions. Let $p$ be the basepoint of $\altcomplex^n$. The filtration $\altcomplex^n_{\bullet}$ can be  described as follows: \\
$\altcomplex^{n}_0 = p$, and for $i\in \{1,\dots,n\}$,
$\altcomplex^{n}_i = S^1 \vee \dots  \vee S^i$,
where for each index $i$, a cell of dimension $i$ is added with the attaching map that takes the boundary of the $i$-cell to the basepoint $p$. The cohomology algebra of $\altcomplex^n$ is trivial in the sense that $x \smile y = 0$ for all $x, y \in \mathsf{H}^{\ast}(\altcomplex)$.

\smallskip

Standard persistence cannot distinguish $\altcomplex^n_{\bullet}$ from $\rpn_{\bullet}$ since they have the same standard persistence barcode. Persistent cup length for $\rpn_{\bullet}$ and  $\cpn_{\bullet}$ for all intervals $[i,j]$ with $n \geq i \geq 1$ is equal to $i$, and hence persistent cup length cannot disambiguate these filtrations.

\smallskip

Finally, persistent cup modules can tell apart $\altcomplex^n_{\bullet}$,  $\rpn_{\bullet}$ and $\cpn_\bullet$ as their cup module barcodes are different. This follows from the fact that the degrees of the generator of the cohomology algebras of $\rpn_{\bullet}$ and $\cpn_\bullet$ are different.
    
\end{example}

In Section~\ref{sec:alg} and \ref{sec:alg-kproduct}, we show how to compute the persistent $k$-cup modules for all $k \in \{2, \dots, d\}$ in $O(d n^4)$ time, where $d$ denotes the dimension of the filtered complex, and $n$ denotes its size.
Moreover, since the persistent cup length of a filtration can be obtained as a byproduct of cup modules computation~\cite{contessoto}, we get an efficient algorithm to compute this invariant as well. 
Our approach for computing barcodes of persistent $k$-cup modules involves computing the image persistence of the cup product viewed as a map from the tensor product of the cohomology vector space to the cohomology vector space itself. This approach requires careful bookkeeping of restrictions of  cocycles as one processes the simplices in the reverse filtration order. Algorithms for computing image persistence have been studied earlier by Cohen-Steiner et al.~\cite{cohenimage} and recently by Bauer and Schmahl~\cite{bauerschmahlsocg}. However, the algorithms in \cite{cohenimage,bauerschmahlsocg} work only for monomorphisms of filtrations making them inapplicable to our setting. 

In \Cref{sec:partition}, we introduce a new invariant called the partition modules of the cup product which is more discriminative than the  $k$-cup modules. We observe that this invariant is stable for $\ripsfull$ and $\cechfull$ filtrations (\Cref{sec:stab}), and we devise an algorithm that computes all the partition modules in $O(c(d)n^4)$ where $c(d)$ is subexponential in $d$ as shown in \Cref{sec:algpart}.

\section{Background and preliminaries}
Througout, we use $n$ to denote the size of the filtered complex $\complex$, $[n]$ to denote the set $\{1,2,\dots,n\}$ and $I$ to denote the set $\{0,1,2,\dots,n\}$. 

\subsection{Persistent cohomology} \label{sec:perscoh}

In this paper, we work with mod-2 cohomology.We refer the reader to \cite{MR1867354,hausmann2014mod} for topological preliminaries.
 Let $P$ denote a  poset category such as $\mathbb{N}$, $\mathbb{Z}$, or $\mathbb{R}$, and $\Simp$ denote the category of simplicial complexes.
 \begin{figure}[htbp]
\centerline{\includegraphics{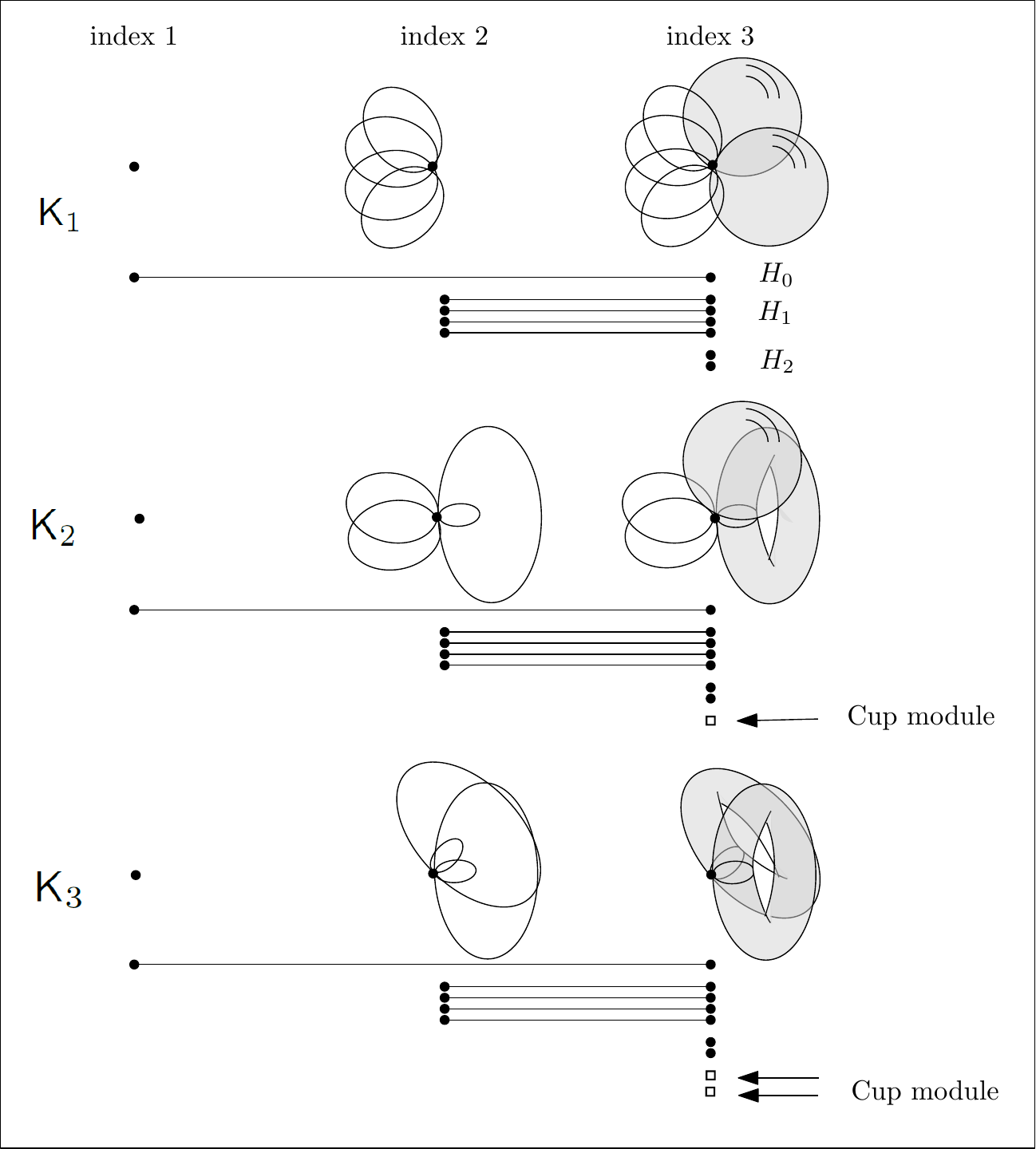}}
\caption{Example 1 Persistent cup modules distinguishes all three cellular filtrations.}
\label{fig:prodpers}
\end{figure}
A \emph{$P$-indexed filtration} is a functor $\mathcal{F}: P \to \Simp$ such that $\mathcal{F}_s \subseteq \mathcal{F}_t$ whenever $s\leq t$.
A $P$-indexed \emph{persistence module} $V_{\bullet}$  is a functor from a poset category $P$ to the category of (graded) vector spaces. The morphisms $\psi_{s,t}: V_s \to V_t$ for $s \leq t$ are referred to as \emph{structure maps}.
We assume it to be of \emph{finite type}, that is, $V_{\bullet}$ is
pointwise finite dimensional and all morphisms $ \psi_{s,t}$ for $s\leq t$ are isomorphisms outside a finite subset of $P$. A $P$-indexed module $W$ is a \emph{submodule} of $V$ if $W_{s}\subset V_{s}$ for all $s\in P$ and the structure maps  $W_s \to W_t$ are restrictions of $\psi_{s,t}$ to $W_s$.

A persistence module $V_{\bullet}$ defined on a 
totally ordered set such as $\mathbb{N}$, $\mathbb{Z}$, or $\mathbb{R}$ decomposes uniquely up to isomorphism
into simple modules called \emph{interval modules} whose structure maps are identity
and the vector spaces have dimension one. The support of these interval modules collectively constitute
what is called the barcode of $V_{\bullet}$ and denoted by $B(V_{\bullet})$.
 
 When we have a filtration $\mathcal{F}$ on $P$ where the complexes change
only at a finite set of values $a_1<a_2<\ldots<a_n$,  we can reindex the filtration with integers, and refine it so that only one simplex is added at every index. Reindexing and refining in this manner one can obtain a simplex-wise filtration of the final simplicial complex $\complex$
defined on an indexing set with integers.
For the remainder of the paper, we assume that the original filtration on $P$ is simplex-wise to begin with. This only simplifies our presentation, and we do not lose generality. With
this assumption, we obtain a filtration indexed on
$I$ after writing $\complex_{a_i}=\complex_i$,
\[ \complex_{\bullet}: \emptyset=\complex_0 \hookrightarrow  \complex_1  \hookrightarrow 
 \dots \hookrightarrow  \complex_n = \complex. \]

Applying the functor $\mathsf{C}^{\ast}$, we obtain a persistence module $\mathsf{C}^{\ast}(\complex_\bullet)$ of cochain complexes  whose structure maps are cochain maps defined by restrictions induced by inclusions:
\[\mathsf{C}^{\ast}(\complex_\bullet): \mathsf{C}^{\ast}(\complex_n) \to  \mathsf{C}^{\ast}(\complex_{n-1})\to \dots \to
\mathsf{C}^{\ast}(\complex_{0}),
\]
and applying the functor $\mathsf{H}^{\ast}$, we get a persistence module $\mathsf{H}^{\ast}(\complex_\bullet)$ of 
 graded cohomology vector spaces whose structure maps are linear maps induced by the above-mentioned restrictions:
\[\mathsf{H}^{\ast}(\complex_\bullet): \mathsf{H}^{\ast}(\complex_n) \to  \mathsf{H}^{\ast}(\complex_{n-1}) \to\dots \to
\mathsf{H}^{\ast}(\complex_{0}).
\]

For simplifying the description of the algorithm, we work with $I^{\op}$-indexed modules $\mathsf{H}^{\ast}(\complex_\bullet)$ and $\mathsf{C}^{\ast}(\complex_\bullet)$. 
    The barcode $B(M)$ (see section~\ref{sec:barcode}) of a finite-type $P^{\op}$-module $M$ can be obtained from the barcode $B(N)$ of its associated $I^{\op}$-module $N$ by writing the interval $(j,i] \in B(N)$ for $j<i<n$ as $[a_{j+1}, a_{i+1}) \in B(M)$, 
    and the interval $(j,n] \in B(N)$ as  $[a_{j+1}, \infty) \in B(M)$. In this convention, we refer to $i$ (or $n$) as a \emph{birth index}, $j$ as a \emph{death index}, and intervals of the form $(j,n]$ as \emph{essential bars}.

\begin{definition}[Restriction of cocycles]
   For a filtration $\complex_{\bullet}$, if $\zeta$ is a cocycle in complex $\complex_{b}$, but ceases to
be a cocycle at $\complex_{b+1}$, then $\zeta^{i}$ is defined as $\zeta^{i}=\zeta\cap \mathsf{C}^{\ast}(\complex_{i})$
for $i\leq b$, and in this case, we say that $\zeta^{i}$ is the restriction of
$\zeta$ to index $i$. For $i>b$ , $\zeta^{i}$ is set to the
zero cocycle. 
\end{definition}

\begin{definition}[Persistent cohomology basis]
    Let $\Omega_\complex =\left\{ \zeta_{\mathbf{i}}\mid\mathbf{i}\in B(\mathsf{H}^{\ast}(\complex_{\bullet}))\right\} $
be a set of cocycles,
where for every $\mathbf{i}=(d_{i},b_{i}]$, $\zeta_{\mathbf{i}}$
is a cocycle in $\complex_{b_{i}}$ but no more a cocycle in $\complex_{b_{i}+1}$. If for every index $j \in [n]$, the
cocycle classes $\left\{ [\zeta_{\mathbf{i}}^{j}]\mid\zeta_{\mathbf{i}}\in\Omega_\complex \right\} $
form a basis for $\mathsf{H}^{\ast}(\complex_{j})$, then we say that $\Omega_\complex $ is a
persistent cohomology basis for $\complex_{\bullet}$, and the cocycle
$\zeta_{\mathbf{i}}$ is called a representative cocycle for the interval
$\mathbf{i}$. If $b_i = n$, $[\zeta_{\mathbf{i}}]$ is called an \emph{essential class}.
\end{definition}

\subsection{Simplicial cup product}

Simplicial cup products connect cohomology groups across degrees.
Let $\prec$ be an arbitrary but fixed total order on the vertex set  of $\complex$.
Let $\xi$ and $\zeta$ be cocycles of degrees $p$ and $q$ respectively.
The cup product of $\xi$ and $\zeta$ is the $(p+q)$-cocycle $\xi \smile  \zeta $ whose
evaluation on any $(p+q)$-simplex $\sigma = \{v_0,\dots,v_{p+q}\}$ is given by
\begin{equation} \label{eq:simpcup}
    (\xi \smile \zeta)(\sigma)=\xi(\left\{ {v_{0},...,v_{p}}\right\}) \cdot \zeta(\left\{ {v_{p},\dots,v_{p+q}}\right\}).
\end{equation}

This defines a map $\smile : \mathsf{C}^p(\complex) \times \mathsf{C}^q(\complex) \to \mathsf{C}^{p+q}(\complex)$, which assembles to give a map 
$\smile : \mathsf{C}^\ast(\complex) \times \mathsf{C}^\ast(\complex) \to \mathsf{C}^{\ast}(\complex)$ for the cochain complex $\mathsf{C}^\ast(\complex)$.
Using the fact that $\delta(\zeta\smile \xi)=\delta\xi \smile \zeta+\xi\smile \delta\zeta$, it follows that $\smile $ induces a map $\smile : \mathsf{H}^\ast(\complex) \times \mathsf{H}^\ast(\complex) \to \mathsf{H}^{\ast}(\complex)$.
It can be shown that the map $\smile$ is independent of the ordering $\prec$.

Using the universal property for tensor products and linearity, the bilinear maps for   
 \[\smile:\Chain^p(\complex)\times\Chain^q(\complex)\to\Chain^{p+q}(\complex) \text{ \,\, assemble to give a linear map \,\,}
     \smile:\Chain^*(\complex)\otimes\Chain^*(\complex)\to\Chain^{*}(\complex).
\] 
and the bilinear maps for   
 \[\smile:\Homol^p(\complex)\times\Homol^q(\complex)\to\Homol^{p+q}(\complex) \text{ \,\, assemble to give a linear map \,\,}
     \smile:\Homol^*(\complex)\otimes\Homol^*(\complex)\to\Homol^{*}(\complex).
\]

Finally, we state two well-known facts about cup products that are used throughout.

\begin{theorem}[Commutativity~\cite{hausmann2014mod}]
    $[\xi] \smile [\zeta] = [\zeta] \smile [\xi]$ for all $[\xi],[\zeta] \in \Homol^{\ast}(\complex)$.
\end{theorem}

\begin{theorem}[Functoriality of the cup product~\cite{hausmann2014mod}]
    Let $f: \complex \to \altcomplex$ be a simplicial map and let $f^{\ast}: \Homol^{\ast}(\altcomplex) \to \Homol^{\ast}(\complex)$ be the induced map on cohomology. Then, $f^{\ast}([\xi] \smile [\zeta]) = f^{\ast}([\xi]) \smile f^{\ast}([\zeta])$ for all $[\xi], [\zeta] \in \Homol^{\ast}(\complex)$.
\end{theorem}

\subsection{Image persistence}

 The category of persistence modules
is abelian since the indexing category $P$ is small and the category of vector spaces is abelian. Thus, kernels,
cokernels, and direct sums are well-defined.
 Persistence modules obtained as images, kernels and cokernels of morphisms were first studied  in \cite{cohenimage}. In this section, we provide a brief overview of image persistence modules.

 Let $\mathsf{C}_{\bullet}$ and $\mathsf{D}_{\bullet}$ be two persistence modules of cochain complexes: 
 \[
 \mathsf{C}^{\ast}_n \xrightarrow{\varphi_{n}}  \mathsf{C}^{\ast}_{n-1} \xrightarrow{\varphi_{n-1}} \dots \xrightarrow{\varphi_{1}}
\mathsf{C}^{\ast}_{0}
\quad\quad
 and 
\quad\quad
  \mathsf{D}^{\ast}_n \xrightarrow{\psi_{n}}  \mathsf{D}^{\ast}_{n-1} \xrightarrow{\psi_{n-1}} \dots \xrightarrow{\psi_{1}}
\mathsf{D}^{\ast}_0,
 \]
such that for $0\leq i\leq n$ the graded vector spaces $\mathsf{C}^{*}_i$ and $\mathsf{D}^{*}_i$ (along with the respective coboundary maps)  are cochain complexes, and  the structure maps $\{\varphi_i:\mathsf{C}^{*}_i \to \mathsf{C}^{*}_{i-1} \mid i \in [n] \}$ and $\{\psi_i:\mathsf{D}^{*}_i \to \mathsf{D}^{*}_{i-1} \mid i \in [n] \}$ are cochain maps.
Let $G_{\bullet}: \mathsf{C}_{\bullet}  \to \mathsf{D}_{\bullet} $ be a  \emph{morphism of persistence modules of cochain complexes}, that is,  there exists a set of cochain maps $G_i: \mathsf{C}^{*}_{i} \to \mathsf{D}^{*}_{i}$  $\forall i \in \{0,\dots,n\}$, and the following diagram commutes for every $i\in [n]$.
\begin{figure}[H]
 \[\begin{tikzcd}
	{\mathsf{C}^{*}_i} &&& {\mathsf{D}^{*}_{i}} \\
	\\
	{\mathsf{C}^{*}_{i-1}} &&& {\mathsf{D}^{*}_{i-1}}
	\arrow["{G_i}", from=1-1, to=1-4]
	\arrow["{\varphi_i}"', from=1-1, to=3-1]
	\arrow["{G_{i-1}}"', from=3-1, to=3-4]
	\arrow["{\psi_i}", from=1-4, to=3-4]
\end{tikzcd}\]
\end{figure}

 Applying the cohomology functor  $\Homol^{*}$  to the morphism  $G_{\bullet}\colon \mathsf{C}_{\bullet}\to \mathsf{D}_{\bullet}$ induces another morphism of persistence modules, namely, $\Homol^{*}(G_{\bullet}) \colon \Homol^{*}(\mathsf{C}_{\bullet}) \to \Homol^{*}(\mathsf{D}_{\bullet})$. Moreover, the image $\im \Homol^{*}(G_{\bullet})$ is a  persistence module.
 Like any other single-parameter persistence module, an image persistence module decomposes uniquely into intervals called its \emph{barcode}~\cite{zomorodian2004computing}.

 As noted in \cite{bauerschmahlsocg}, a natural  strategy for computing the image of $\Homol^{*}(G_{\bullet})$ is to  write it as
\[
\im \Homol^{*}(G_{\bullet})\cong \frac{G_{\bullet}(\Cycl^{*}(\mathsf{C}_{\bullet}))}{G_{\bullet}(\Cycl^{*}(\mathsf{C}_{\bullet}))\cap \Bdry^{*}(\mathsf{D}_{\bullet})},
\]
where the $i$-th terms for the numerator and the denominator are given  respectively by
$(G_{\bullet}(\Cycl^{*}(\mathsf{C}_{\bullet})))_i=G_{i}(\Cycl^{*}(\mathsf{C}_{i})) \mbox{ and }
(G_{\bullet}(\Cycl^{*}(\mathsf{C}_{\bullet}))\cap \Bdry^{*}(\mathsf{D}_{\bullet}))_{i} = G_{i}(\Cycl^{*}(\mathsf{C}_{i}))\cap \Bdry^{*}(\mathsf{D}_{i}).
$

\paragraph*{Tensor product image persistence.}
Consider the following map given by cup products
 \begin{equation} \label{eqn:tensormaptwo}
     \smile_{\bullet}:\Chain^*(\complex_{\bullet})\otimes\Chain^*(\complex_{\bullet})\to\Chain^{*}(\complex_{\bullet}).
 \end{equation}
Taking $G_{\bullet}=\smile_{\bullet}$ in the definition of
image persistence, we get
a persistence module, denoted by $\im \Homol^{*}(\smile \complex_{\bullet})$, which is the same as the persistent cup module introduced in~\cite{cuparxivtwo}. Whenever the underlying filtered complex is clear from the context, we  use the shorthand notation   $\im \Homol^{*}(\cupmap)$ instead of $\im \Homol^{*}(\smile \complex_{\bullet})$. Our aim is to compute its barcode denoted by
 $B(\im \Homol^{*}(\cupmap))$.

\subsection{Barcodes}
\label{sec:barcode}
.

Let $\mathsf{K}_\bullet$ denote a  filtration on the index set
$I = \{0,1,\ldots, n\}$. Assume that $\mathsf{K}_\bullet$ is simplex-wise, that is,
$\complex_i\setminus \complex_{i-1}$ is a single simplex. 
Consider the persistence module $\mathsf{H}^\ast_{\bullet}$ obtained by applying the cohomology functor $\mathsf{H}^\ast$ on the filtration $\mathsf{K}_\bullet$, that is,
$\mathsf{H}^\ast_i=\mathsf{H}^\ast(\mathsf{K}_i)$. The structure maps $\{\cohommap_i:\mathsf{H}^\ast(\mathsf{K}_i) \to \mathsf{H}^\ast(\mathsf{K}_{i-1}) \mid i \in [n] \}$
for this module are induced
by the cochain maps $\{\cochainmap_i:\mathsf{C}^{*}(\mathsf{K}_i) \to \mathsf{C}^{*}(\mathsf{K}_{i-1}) \mid i \in [n] \}$. Since $\mathsf{K}_\bullet$ is
simplex-wise, each linear map $\cohommap_i$ is either injective with
a cokernel of dimension one, or surjective with a kernel of dimension one, but
not both. Such a persistence module $\mathsf{H}^\ast_{\bullet}$ decomposes into interval modules supported on a unique set of intervals, namely the barcode of  $\mathsf{H}^\ast_{\bullet}$ written as $B(\mathsf{H}^\ast_{\bullet})=\{(d_i,b_i]\,|\, b_i\geq d_i, b_i,d_i\in I\}$. 
Notice that since $I$ is the indexing poset  of $\complex_\bullet$,  $I^{\op}$  is the indexing poset of  $\mathsf{H}^\ast_\bullet$.
For $r>s$, we 
define $\cohommap_{r,s}= \cohommap_{s+1}\circ\cdots\circ\cohommap_{r-1}\circ\cohommap_r$ and $\cochainmap_{r,s}=\cochainmap_{s+1}\circ\cdots\circ \cochainmap_{r-1}\circ\cochainmap_r$.

\begin{remark}\label{rem:defmap}
Since $\productmodulek$ is a submodule of $\ordmodule$, the
 structure maps  of $\productmodulek$ for every $i \in I$, namely, $\im \Homol^{*}(\smile_{i}) \to \im \Homol^{*}(\smile_{i-1})$  are given by restrictions of $\cohommap_i$ to $\im \Homol^{*}(\smile_{i})$.
\end{remark}

\begin{definition}
    For any $i\in \{0,\dots,n\}$, a nontrivial cocycle  $\zeta \in \mathsf{Z}^{\ast}(\complex_i)$  is said to be a \emph{product cocycle of $\complex_i$} if $[\zeta] \in\im \Homol^{*}(\smile_{i})$.
\end{definition}

\begin{proposition} \label{prop:mainstruct}
    For a filtration $\complex_\bullet$, the birth indices of $B(\productmodulek)$ are a subset of the birth indices of $B(\ordmodule)$, and the death indices of $B(\productmodulek)$ are a subset of the death indices of $B(\ordmodule)$. 
\end{proposition}
\begin{proof}
   Let $(d_i,b_i]$ and $(d_j,b_j]$ be (not necessarily distinct) intervals in $B(\ordmodule)$, where $b_j \geq b_i$. Let $\xi_i$ and $\xi_j$ be  representatives for $(d_i,b_i]$ and $(d_j,b_j]$ respectively.
    If $\xi_i \smile \xi_j^{b_i}$ is trivial, then by the functoriality of cup product, $\cochainmap_{b_i,r}(\xi_i \smile \xi_j^{b_i}) = \cochainmap_{b_i,r}(\xi_i) \smile \cochainmap_{b_i,r}(\xi_j^{b^i}) = \xi_i^{r} \smile  \xi_j^{r}$ is trivial  $\forall r < b_i$. Writing contrapositively, if $\exists r< b_i$ for which $\xi_i^{r} \smile \xi_j^{r} $ is nontrivial, then $\xi_i \smile \xi_j^{b_i}$ is nontrivial.
    Noting that $\im \Homol^{*}(\smile_{\ell})$ for any $\ell \in \{0,\dots,n\}$ is generated by $\{[\xi_i^\ell] \smile [\xi_j^
    \ell] \mid \xi_i, \xi_j 
\in \Omega_\complex \}$, it follows that an index $b$ is the birth index of a bar in $B(\productmodulek)$ only if it is  the birth index of a bar in $B(\ordmodule)$, proving the first claim. 

Let $\Omega'_{j+1} = \{ [\tau_1],\dots, [\tau_k]\}$ be a basis for $\im \Homol^{*}(\smile_{j+1})$. Then, $\Omega'_{j+1}$ extends to a basis $\Omega_{j+1}$ of $\Homol^{\ast}(\complex_{j+1})$.
If $j$ is not a death index of $B(\ordmodule)$, then $\cochainmap_{j+1}(\tau_1), \dots, \cochainmap_{j+1}(\tau_k)$ are all nontrivial and linearly independent. From \Cref{rem:defmap}, it follows that $j$ is not a death index of $B(\productmodulek)$, proving the second claim. 
 \end{proof}

\begin{corollary} \label{cor:onedie}
For a filtration $\complex_\bullet$, if $d$ is a death index of $B(\productmodulek)$, then at most one bar of $B(\productmodulek)$ has death index $d$.
\end{corollary}
\begin{proof}
Using the fact that if the rank of a linear map $f: V_1 \to V_2$ is $\dim V_1 -1$, then the rank of $f |_{W_1}$ for a subspace $W_1 \subset V_1$ is at least $\dim W_1 -1$, from \Cref{rem:defmap} it follows that
if $\dim \Homol^{*} (\complex_{d}) = \dim \Homol^{*} (\complex_{d+1})  - 1$, then
 \[\dim (\im \Homol^{*}(\smile_{d})) + 1 \geq \dim (\im \Homol^{*}(\smile_{d+1})) \geq \dim (\im \Homol^{*}(\smile_{d}))\quad\text{  proving the claim.} \qedhere \]
\end{proof}

   \begin{remark}
  The persistent cup module is a submodule of the original persistence module. 
  Let $ \dim (\im \Homol^{p}_i)$ denote $\dim (\im \Homol^{p}(\smile_{i}))$.
 In the barcode $B(\productmodulek)$, if $\complex_i = \complex_{i-1} \cup \{\sigma^p\}$, then either (i) $\dim (\im \Homol^{p}_i) > \dim (\im \Homol^{p}_{i-1})$,  or (ii) $\dim (\im \Homol^{p-1}_i) < \dim (\im \Homol^{p-1}_{i-1})$, or (iii) there is no change: $\dim (\im \Homol^{p}_i) = \dim (\im \Homol^{p}_{i-1})  $ and  $\dim (\im \Homol^{p-1}_i) = \dim (\im \Homol^{p-1}_{i-1})$.
 The decrease (increase) in persistent cup modules happens only if there is a decrease (increase) in ordinary cohomology. 
 Multiple bars of $B(\productmodulek)$ may have the same birth index. 
 But, if $i$ is a death index, then Corollary~\ref{cor:onedie} says that
 it is so for at most one bar in $B(\productmodulek)$.
 \end{remark}

\section{Algorithm: barcode of persistent cup module} 
\label{sec:alg}
Our goal is to compute the barcode of 
$\productmodulek$, which being an image module is a submodule
of  $\mathsf{H}^\ast(\mathsf{K}_\bullet)$. The vector space $\im \mathsf{H}^\ast(\cupmapi)$ is a subspace of the cohomology vector space 
$\mathsf{H}^\ast(\mathsf{K}_i)$. Let us call this subspace  
the \emph{cup space} of $\mathsf{H}^\ast(\mathsf{K}_i)$. Our algorithm
keeps track of a basis of this cup space as it processes
the filtration in the reverse order. This backward processing
is needed because the structure maps  
 between the cup spaces 
are induced by
restrictions  $\varphi_{j,i} \colon \mathsf{C}^{*}(\complex_j) \to \mathsf{C}^{*}(\complex_i)$ that are, in turn, induced by inclusions $\mathsf{K}_j\supseteq \mathsf{K}_i$, $i\leq j$.
In particular, a cocycle/coboundary in $\complex_j$ is taken to its restriction in $\complex_i$ for $i\leq j$. 
Our algorithm keeps track of the birth and death of the cocycle classes in the cup spaces as it proceeds through the restrictions in the reverse filtration order. 
We maintain a basis of nontrivial product cocycles in a matrix $\matrixb$ whose classes $S$ form a basis for the cup spaces. In particular, cocycles in $\matrixb$ are born and die with birth and death of the elements in cup spaces.

A cocycle class from $\mathsf{H}^\ast(\mathsf{K}_i)$ may
enter the cup space $\im \mathsf{H}^\ast(\cupmapi)$ signalling a birth or may leave (become zero) the cohomology vector space and hence the cup space signalling a death. Interestingly,
multiple births may happen, meaning that multiple independent 
cocycle classes 
may enter the cup space, whereas at most a single class can die because of \Cref{cor:onedie}.
To determine which class from the cohomology vector space 
enters the cup space and which one leaves it, we make use of the barcode of  $\mathsf{H}^\ast(\mathsf{K}_\bullet)$. 
However, the classes of 
the bases maintained in $\matrixh$ do not directly provide bases for the cup spaces. Hence, we need to compute and maintain $\matrixb$ separately, of course, with the
help of $\matrixh$.

Let us consider the case of birth first. Suppose that a cocycle $\xi$ at degree
$p$ is born at index $k = b_i$ for $\Homol^\ast(\complex_\bullet)$. With $\xi$, a set of product cocycles are born in some of the degrees $p+q$ for $q\geq 1$. 
To detect them, we first compute a set of candidate 
cocycles by taking the cup product of cocycles $\xi \smile \zeta$, for all  cocycles $\zeta \in \matrixh$ at $b_i$ 
which can potentially augment the  basis maintained in $\matrixb$. The ones among the
candidate cocycles whose classes are independent
w.r.t. the current  basis maintained in $\matrixb$ are determined to be born at $b_i$. 
Next, consider the case of death.
A product cocycle $\zeta$ in degree $r$ ceases to exist if it becomes linearly dependent of other product cocycles. This can happen only if the dimension of $\Homol^r(\complex_\bullet)$ itself has reduced under the structure map going from $k+1$ to $k$.

It suffices to check if any of the nontrivial cocycles in $\matrixb$ have become linearly dependent or trivial
after applying  restrictions. In what follows, we use $\deg(\zeta)$ to denote the degree of a cocycle $\zeta$.\\

\noindent
{\bf Algorithm} {\sc CupPers} ($\mathsf{K}_\bullet$)\label{algorithm:cuppers}
\begin{itemize}
    \item Step 1. Compute barcode $\barcode =\{(d_i,b_i]\}$ of $\Homol^\ast(\complex_\bullet)$
    with representative cocycles $\xi_i$; Let $\matrixh=\{\xi_i  \mid [\xi_i] \text{ essential and } \deg(\xi_i)>0 \}$; Initialize $\matrixb$  with the coboundary matrix $\partial^{\perp}$ obtained by taking transpose of the boundary matrix $\partial$;
    \item Step 2. For $k:=n$ to $1$ do
    \begin{itemize}
        \item Restrict the cocycles in $\matrixb$ and $\matrixh$ to index $k$; 
        \item Step 2.1 For every $i$ with  $k=b_i$ ($k$ is a birth-index)  and $\deg(\xi_i)>0 $  
        \begin{itemize}
            \item Step 2.1.1 If $k\not = n$, update $\matrixh:=[\matrixh \mid \xi_i]$
            \item Step 2.1.2 For every $\xi_j \in \matrixh$ 
            \begin{enumerate}
                \item [i.] If $(\zeta\leftarrow \xi_i\smile \xi_j)\not = 0$ and $\zeta$ is
                independent in $\matrixb$, then $\matrixb:=[\matrixb \mid \zeta]$ with column $\zeta$ annotated
                as $\zeta\cdot {\rm birth}:=k$ and  $\zeta\cdot {\rm rep@birth}:= \zeta$
            \end{enumerate}
        \end{itemize}
        \item Step 2.2 If $k=d_i$ ($k$ is a death-index) for some $i$ and $\deg(\xi_i)>0 $ then 
        \begin{itemize}
         \item Step 2.2.1 Reduce $\matrixb$ with left-to-right column additions 
         \item Step 2.2.2 If a nontrivial cocycle $\zeta$ is zeroed out, remove $\zeta$ from $\matrixb$, generate the
         bar-representative pair $\{(k,\zeta\cdot {\rm birth}],\zeta\cdot {\rm rep@birth}\}$
         \item Step 2.2.3 Update $\matrixh$ by removing the column $\xi_i$
        \end{itemize}
    \end{itemize}
\end{itemize}
Algorithm {\sc CupPers} describes this algorithm with a pseudocode. First, in Step 1, we
compute the barcode of the cohomology persistence module $\mathsf{H}^\ast(\mathsf{K}_\bullet)$
along with a persistent cohomology basis.
This can be achieved in $O(n^3)$ time using either the annotation
algorithm~\cite{boissonnat2013compressed,DW22} or the pCoH algorithm~\cite{de2011dualities}.
The basis $H$ is maintained with the matrix $\matrixh$  whose columns are cocycles represented as the support
vectors on simplices. The matrix $\matrixh$ is initialized with all cocycles $\xi_i$ that
are computed as representatives of the bars $(d_i,b_i]$ for the module $\mathsf{H}^\ast(\mathsf{K}_\bullet)$ which get born at the first (w.r.t. reverse order)
complex $\mathsf{K}_n=\mathsf{K}$. The matrix $\matrixb$ is initialized with the coboundary matrix $\partial^{\perp}$ with standard cochain basis. Subsequently, nontrivial cocycle vectors are added to $\matrixb$. The classes of the nontrivial cocycles in matrix $\matrixb$ form a basis $\basisB$ for the cup space at any point in the course of the algorithm.

In Step 2, we process cocycles in the reverse filtration order. At each index $k$, we do
the following. If $k$ is a birth index for a bar $(-,b_i]$ (Step 2.1), that is, $k=b_i$ for
a bar with representative $\xi_i$ in the barcode of
$\mathsf{H}^\ast(\mathsf{K}_\bullet)$, first we augment $\matrixh$ with $\xi_i$ to keep it current
as a basis for the vector space $\mathsf{H}^\ast(\mathsf{K}_k)$ (Step 2.2.1). Now,
a new bar for
the persistent cup module can potentially
be born at $k$. To determine this, we take the cup product of $\xi_i$ with all
cocycles in $\matrixh$ and check if the cup product cocycle is non-trivial and is
independent of the cocycles in $\matrixb$. If so, a product cocycle
is born at $k$ that is added to $\matrixb$ (Step 2.1.2). 
To check this independence, we need $\matrixb$ to have current coboundary basis
along with current nontrivial product cocycle basis $S$ that are both updated with restrictions. Note that we need a for loop in Step 2.1 because at $k=n$, there can be multiple births in $\mathsf{H}^\ast(\mathsf{K}_\bullet)$. 

\begin{remark}
    Restrictions in $\matrixh$ and $\matrixb$ are implemented by zeroing out the corresponding row associated to the simplex $\sigma_i$ when
we go from $\mathsf{K}_i$ to $\mathsf{K}_{i-1}$ and $\mathsf{K}_i\setminus \mathsf{K}_{i-1}=\{\sigma_i\}$.
\end{remark}

If $k$ is a death index (Step 2.2), potentially the class of a product cocycle from $\matrixb$
can be a linear combination of the classes of other product cocycles after $\matrixb$ has been
updated with restriction. We reduce $\matrixb$ with left-to-right column additions and
detect the column that is zeroed out (Step 2.2.1). If the column $\zeta$ is zeroed out,
the class $[\zeta]$ dies at $k$ and we generate a bar with death index $k$ and birth index
equal to the index when $\zeta$ was born (Step 2.2.2). Finally, we update $\matrixh$ by removing the
column for $\xi_i$ (Step 2.2.3).

\smallskip

\subsection{Rank functions and barcodes}
Let $P\subseteq \mathbb{Z}$ be a finite set with induced poset structure
from $\mathbb{Z}$. Let $\Int(P)$ denote the set of all intervals in $P$.
Recall that $P^{\op}$ denotes the opposite poset category.
Given a $P^{\op}$-indexed persistence module $V_\bullet$, the rank function 
$\rk_{V_\bullet}: \Int(P)\rightarrow \mathbb{Z}$ assigns to each interval
$I=[a,b]\in \Int(P)$ the rank of the linear map $V_b\rightarrow V_a$.
It is well known that (see~\cite{CEH07,EH10}) the barcode of $V_\bullet$ viewed as a function 
$\Dgm_{V_\bullet}: \Int(P)\rightarrow \mathbb{Z}$ can be obtained from the rank function
by the inclusion-exclusion formula:
\begin{equation}
\Dgm_{V_\bullet}([a,b])= \rk_{V_\bullet}[a,b]-\rk_{V_\bullet}[a-1,b] + \rk_{V_\bullet}[a,b+1]-\rk_{V_\bullet}[a-1,b+1]
\label{eq:mobius}
\end{equation}

    To prove the correctness of {Algorithm} {\sc CupPers}, we use the following elementary fact.
\begin{fact} \label{fa:rankfact}
A class that is born at an index $\geq b$ dies at $a$
 iff $\rk_{V_\bullet}([a,b]) < \rk_{V_\bullet}([a+1,b])$.
\end{fact}

\subsection{Correctness of {\bf Algorithm} {\sc CupPers}}

\begin{theorem} \label{thm:maincorrect}
Algorithm {\sc CupPers}  computes the barcode of  the persistent cup module.
\end{theorem}
\begin{proof}
In what follows, we abuse notation by denoting the restriction at index $k$  of  a cocycle $\zeta$ born at $b$ also by the symbol $\zeta$. That is, index-wise restrictions are always performed, but not always explicitly mentioned. 
We use $\{\xi_i\}$ to denote cocycles in the persistent cohomology basis computed in Step 1.
The proof uses induction to show  that  for an arbitrary birth index $b$ in $B(\ordmodule)$, if all
bars for the persistent cup module with birth indices $b' > b$ 
 are correctly computed, then
the  bars beginning with $b$ are also correctly computed.

To begin with we  note that in Algorithm {\sc CupPers}, as a consequence of \Cref{prop:mainstruct}, we need to check if an index $k$ is a birth (death) index  of $B(\productmodulek)$ only when it is a birth 
 (death) index of 
$B(\ordmodule)$. Also, from \Cref{cor:onedie}, we know that at most one cycle dies at a death index of $B(\productmodulek)$ (justifying Step 2.2.2).

We now introduce some notation.
In what follows, we  denote the persistent cup module by $V_\bullet$. 
For a birth index $b$, let $\basisB_{b}$ be the cup space at index $b$. Let $C_{b}$ be the vector space of the product cocycle classes created at index $b$.
In particular, the classes in  $C_{b}$ are linearly independent of classes in $\basisB_{b+1}$.
For a birth index $b<n$, $\basisB_{b}$ can be written as  a direct sum $\basisB_{b} = \basisB_{b+1} \oplus C_{b}$. 
For index $n$, we set $\basisB_n = C_n$. Then, for a birth index $b \in \{0,\dots,n\}$,
$C_b$ is a subspace of $\Homol^{\ast}(\complex_b)$.   $C_b$ can be written as:

\begin{equation*} 
C_b = 
\begin{cases}  \langle [\xi_i] \smile [\xi_j] \mid \xi_i, \xi_j \text{ are essential cocycles of } \Homol^\ast(\complex_\bullet) \rangle & \text{ if } b=n \\ 
 \langle [\xi_i] \smile [\xi_j] \mid \xi_i\text{ is born at $b$,} \text{ and }\xi_j \text{ is  born at an index } \geq b \rangle & \text{ if } b<n
\end{cases} 
\end{equation*}

For a birth index $b$, let $\matrixc_{b}$ be the submatrix of $\matrixb$ formed by representatives whose classes generate $C_b$, which augments $\matrixb$ in Step 2.1.2 (i) when $k=b$ in the \textbf{for} loop. The cocycles in $\matrixc_{b}$ are maintained for $k\in 
\{b,\dots,1\}$ via subsequent restrictions to index $k$.
Let $\matrixb_{b}$   be the submatrix of $\matrixb$ containing representative product cocycles that are born at index $\geq b$. Clearly, $\matrixc_{b}$ is a submatrix of $\matrixb_{b}$ for $b<n$, and $\matrixc_{n} = \matrixb_{n}$.

Let $\DP_b$ be the set of filtration indices for which the cocycles in $\matrixc_b$ become successively linearly dependent to other cocycles in $\matrixb_b$.
That is, $d \in \DP_b$ if and only if there exists a cocycle $\zeta$ in $\matrixc_b$ such that $\zeta$ is independent of all cocycles to its left in matrix $\matrixb$ at index $d+1$, but $\zeta$ is either trivial or a linear combination of cocycles to its left at index $d$.

For the base case, we show that the death indices of the essential bars are correctly computed.
First, we observe that for all $d \in \DP_n$, $\rk_{V_\bullet}([d,n]) = \rk_{V_\bullet}([d+1,n]) - 1$. 
Using Fact~\ref{fa:rankfact}, it follows that the algorithm computes the correct barcode for $\productmodulek$ only if the indices in $\DP_n$ are the respective death indices for the essential bars. Since the leftmost columns of $\matrixb$ are  coboundaries from $\partial^{\perp}$ followed by cocycles from $\matrixc_n$, and since we perform only left-to-right column additions in Step 2.2.1 to zero out cocycles in  $\matrixc_n$, the base case holds true. By (another) simple inductive argument, it  follows that the computation of indices in $\DP_n$  does not depend on the specific ordering of representatives within $\matrixc_n$.

Let $b<n$ be a birth index in $B(\ordmodule)$.
For induction hypothesis, assume that for every birth index $b' > b$ the indices in $\DP_{b'}$ are the respective death indices of the bars of $\productmodulek$ born at $b'$.
By construction, the cocycles $\{\zeta_1, \zeta_2, \dots\}$ in $\matrixb$ are sequentially arranged by the following rule: If $\zeta_i$ and $\zeta_j$ are two representative product cocycles in $\matrixb$, then $i<j$ if  the birth index $b_i$ of the interval represented by $\zeta_i$ is greater than or equal to the birth index $b_j$ of the interval represented by $\zeta_j$. Then, as a consequence of the induction hypothesis, for a cocycle $\zeta  \in \matrixc_{b} \setminus \matrixb_{b}$,  we assign the correct birth index to the interval represented by $\zeta$ only if $\zeta$ can be written as a linear combination of cocycles to its left in matrix $\matrixb$. 

Now, suppose that at some index $d \in \DP_b$ we can write a cocycle $\zeta$ in submatrix $\matrixc_{b}$ as a linear combination of cocycles to its left in $\matrixb$.
For such a $d \in \DP_b$, $\rk_{V_\bullet}([d,b]) = \rk_{V_\bullet}([d+1,b]) - 1$.  Hence, using Fact~\ref{fa:rankfact}, a birth index $\geq b$ must be paired with $d$. 

However, since $\DP_b \cap \DP_{b'} = \emptyset$ for $b < b'$, it follows from the inductive hypothesis  that the only birth index that can be paired to $d$ is $b$.
Moreover, since we take restrictions of cocycles in $\matrixb$, all cocycles in $\matrixc_{b}$ eventually become trivial or linearly dependent on cocycles to its left in $\matrixb$. So, $\DP_b$ has the same cardinality as the number of cocycles in  $\matrixc_{b}$, and all the bars that are born at $b$ must die at some index in $\DP_b$.
  
As a final remark, it is easy to check that the computation of indices in $\DP_b$  is independent of the specific ordering of representatives within $\matrixb_b$ by a simple inductive argument.
\end{proof}

\paragraph*{Time complexity of {\sc CupPers}.}
Let the input simplex-wise filtration have $n$ additions and hence the complex $\complex$ have
$n$ simplices.
  Step 1 of {\sc CupPers} can be executed in $O(n^3)$ time using algorithms in \cite{boissonnat2013compressed,de2011dualities}. 
  The outer loop
  in Step 2 runs $O(n)$ times. 
  For each death index in Step 2.2,  we perform  left-to-right column additions as done in the standard persistence algorithm to bring the matrix in reduced form.
  Hence, for each death index, Step 2.2 can be performed in $O(n^3)$ time.
  Since there are at most $O(n)$ death indices, the total cost for Step 2.2 in the course of the algorithm is $O(n^4)$.

  Step 2.1 apparently incurs higher cost than Step 2.2. This is because at each birth point, we have to test the
  product of multiple pairs of cocycles stored in $\matrixh$. However, we observe  that there are at most $O(n^2)$ products of pairs of representative cocycles that are each computed and tested for linear independence at most once. In particular, if $\xi_i$ and  $\xi_j$ represent $(d_i,b_i]$  and $(d_j,b_j]$ resp. with $b_i \leq b_j$, then $\xi_i \smile \xi_j$ is computed and tested for  independence iff $b_i>d_j$ and the test happens at  $b_i$. Using \Cref{eq:simpcup}, computing $\xi_i \smile \xi_j$ takes linear time. So the cost of computing the $O(n^2)$ products is $O(n^3)$.
Moreover, since each  independence test takes $O(n^2)$ time with the assumption
  that $\matrixb$ is kept reduced all the time, Step 2.1 can be implemented to run
  in $O(n^4)$ time over the entire algorithm.
  
 Finally, since restrictions of cocycles in $\matrixb$ and $\matrixh$ are computed by zeroing out corresponding rows, the total time to compute restrictions over the course of the algorithm is $O(n^2)$.
    Combining all costs, we get an $O(n^4)$  complexity bound for {\sc CupPers}.

\section{Algorithm: barcode of persistent k-cup modules}
\label{sec:alg-kproduct}
While considering the \emph{persistent $2$-cup modules} (referred to as \emph{persistent cup modules} in \Cref{sec:alg}) is the natural first step, it must be noted that the  invariants thus computed can still be enriched by considering \emph{persistent $k$-cup modules}.
As a next step, we consider image persistence of  the $k$-fold tensor products.

 \paragraph*{Image persistence of $k$-fold tensor product.}
 Consider  image persistence of the map
 \begin{equation} \label{eqn:tensormapfour}
     \cupmap^{k}:\Chain^{*}(\complex_{\bullet})\otimes\Chain^{*}(\complex_{\bullet}) \otimes \dots \otimes\Chain^{*}(\complex_\bullet) \to\Chain^{*}(\complex_{\bullet})
 \end{equation}
  where the tensor product is taken $k$ times. 
 Taking $G_{\bullet} \,\,= \,\,\cupmap^k$ in the definition of
image persistence, we get
the  module $\im \Homol^{*}(\cupmap^k)$ which is same as the
persistent $k$-cup module introduced in \cite{cuparxivtwo}. Our aim is to compute $B(\im \Homol^{*}(\smile^k \complex_\bullet))$ (written as $B(\im \Homol^{*}(\smile^k_\bullet))$ when the complex is clear from the context). Likewise, the degree-wise barcodes $B(\im \Homol^{p}(\cupmap))$ and $B(\im \Homol^{p}(\cupmap^k))$ can also be defined and computed. We omit the details for brevity.

\begin{definition}
    For any $i\in \{0,\dots,n\}$, a nontrivial cocycle  $\zeta \in \mathsf{Z}^{\ast}(\complex_i)$  is said to be an  \emph{order-$k$ product cocycle of $\complex_i$} if $[\zeta] \in\im \Homol^{*}(\smile_{i}^{k})$.
\end{definition}

\subsection{Computing barcode of persistent k-cup modules} 
The order-$k$ product cocycles can be viewed recursively as 
cup products of order-$(k-1)$ product cocycles with another
cocycle. This suggests a recursive algorithm for computing the barcode of  persistent $k$-cup module: compute the barcode of persistent $(k-1)$-cup module recursively and then use that to compute the  barcode of  persistent $k$-cup module
just like the way we computed persistent $2$-cup module
using the bars for ordinary  persistence. In the algorithm
{\sc OrderkCupPers}, we assume that the barcode with representatives 
for $\mathsf{H}^\ast(\mathsf{K}_\bullet)$ has been
 precomputed which is denoted by the pair of sets
$(\{(d_{i,1},b_{i,1}], \{\xi_{i,1}\})$. For simplicity, we assume that this pair is accessed by the
recursive algorithm as a global variable and  is not passed
at each recursion level. At each recursion level $k$, the algorithm
computes the barcode-representative pair denoted as $(\{(d_{i,k},b_{i,k}], \{\xi_{i,k}\})$.
Here, the cocycles $\xi_{i,k}$ are the initial cocycle representatives (before restrictions) for the bars $(d_{i,k},b_{i,k}]$. At the time of their respective births $b_{i,k}$, they are stored in the field $\xi_{i,k} \cdot {\rm rep@birth}$. \\

\noindent
{\bf Algorithm} {\sc OrderkCupPers} ($\mathsf{K}_\bullet$,$k$)\label{algorithm:orderkcuppers}
\begin{itemize}
    \item Step 1. If $k=2$, return the barcode with representatives 
    $\{(d_{i,2},b_{i,2}], \xi_{i,2}\}$ 
    computed by {\sc CupPers} on
    $\mathsf{K}_\bullet$\\
    else $\{(d_{i,k-1},b_{i,k-1}],\xi_{i,k-1}\}$ $\leftarrow$ {\sc OrderkCupPers}($\mathsf{K}_\bullet$, $k-1$)\\
  \hspace*{0.2in} Let $\matrixh=\{\xi_{i,1}  \mid [\xi_{i,1}] \text{ essential \& }\deg(\xi_{i,1})>0 \}$; 
  $\matrixr:=\{\xi_{i,k-1} \mid b_{i,k-1}=n\}$; $\matrixb:=\partial^{\perp}$;
    \item Step 2. For $\ell:=n$ to $1$ do
    \begin{itemize}
        \item Restrict the cocycles in $\matrixb$, $\matrixr$, and $\matrixh$ to index $\ell$; 
        \item Step 2.1 For  every $r$ s.t. $b_{r,1} = \ell \neq n$ (i.e., $\ell$ is a birth-index) and $\deg(\xi_{r,1})>0 $   
        \begin{itemize}
            \item Step 2.1.1 Update $\matrixh:=[\matrixh \mid \xi_{r,1}]$ 
            \item Step 2.1.2 For every  $\xi_{j,k-1} \in \matrixr$ 
            \begin{enumerate}
                \item [i.] If $(\zeta\leftarrow \xi_{r,1}\smile \xi_{j,k-1})\not = 0$ and $\zeta$ is
                independent in $\matrixb$, then $\matrixb:=[\matrixb \mid \zeta]$ with column $\zeta$ annotated
                as $\zeta\cdot {\rm birth}:=\ell$ and $\zeta\cdot {\rm rep@birth}:= \zeta$
            \end{enumerate}
        \end{itemize}

        \item Step 2.2 For all $s$ such that  $\ell=b_{s,k-1}$      
        \begin{itemize}
            \item Step 2.2.1 If $\ell\not = n$, update  $\matrixr:=[\matrixr\mid \xi_{s,k-1}]$
            \item Step 2.2.2 For every  $\xi_{i,1} \in \matrixh$ 
            \begin{enumerate}
                \item [i.] If $(\zeta\leftarrow \xi_{s,k-1}\smile \xi_{i,1})\not = 0$ and $\zeta$ is
                independent in $\matrixb$, then $\matrixb:=[\matrixb \mid \zeta]$ with column $\zeta$ annotated
                as $\zeta\cdot {\rm birth}:=\ell$ and $\zeta\cdot {\rm rep@birth}:= \zeta$
            \end{enumerate}
        \end{itemize}        
        
        \item Step 2.3 If $\ell=d_{i,1}$ (i.e. $\ell$ is a death-index)  and $\deg(\xi_{i,1})>0 $  for some $i$ then 
        \begin{itemize}
         \item Step 2.3.1 Reduce $\matrixb$ with left-to-right column additions 
         \item Step 2.3.2 If a nontrivial cocycle $\zeta$ is zeroed out, remove $\zeta$ from $\matrixb$, generate the
         bar-representative pair $\{(\ell,\zeta\cdot {\rm birth}],\zeta\cdot {\rm rep@birth}\}$
         \item Step 2.3.3 Remove the column $\xi_{i,1}$ from $\matrixh$ 
         \item Step 2.3.4 Remove the column $\xi_{j,k-1}$ from $\matrixr$ if $d_{j,k-1}=\ell$ for some $j$
        \end{itemize}
    \end{itemize}
\end{itemize}

A high-level pseudocode for computing the barcode of  persistent $k$-cup module is given by algorithm {\sc OrderkCupPers}. The algorithm calls itself recursively to generate the
sets of bar-representative pairs for the  persistent $(k-1)$-cup module. As in the case of persistent $2$-cup modules, birth and death indices
of  order-$k$ product cocycle classes are subsets of birth and death indices resp.
of  ordinary persistence. Thus, as before, at each birth index of the cohomology module, we check if the cup product of a representative cocycle (maintained
in matrix $\matrixh$) with a
representative for persistent $(k-1)$-cup module
(maintained in matrix $\matrixr$) generates 
a new cocycle in the barcode for persistent $k$-cup module  (Steps 2.1.2(i), 2.2.2(i)). If so,
we note this birth with the resp. cocycle (by annotating the column) and add it to
the matrix $\matrixb$ that maintains a basis for live order-$k$ product cocycles. At each
death index, we check if an order-$k$ product cocycle dies by checking if the matrix $\matrixb$
loses a rank through restriction (Step 2.3.1). If so, the cocycle in $\matrixb$ that becomes
dependent to other cocycles through a matrix reduction is designated
to be killed (Step 2.3.2) and we note the death of a bar in the $k$-cup module  barcode.
We update $\matrixh$, $\matrixr$ appropriately (Steps 2.3.3, 2.3.4). At a high level,
this algorithm is similar to {\sc CupPers} with the role of $\matrixh$ played by
both $\matrixh$ and $\matrixr$ as they host the cocycles whose products are
to be checked during the birth and the role of $\matrixb$ in both algorithms remains
the same, that is, check if a product cocycle dies or not.

\paragraph*{Correctness and complexity of {\sc OrderkCupPers}.} Correctness can be
established the same way as for {\sc CupPers}. See \Cref{sec:orderk} for a sketch of the proof. For  complexity, observe that we incur a cost from recursive calling
in Step 1 and $O(n^4)$ cost from Step 2 with a similar analysis we did for {\sc CupPers} while noting that there are once again a total of $O(n^2)$ product cocycles to be checked for independence at birth (Steps 2.1 and 2.2).
Then, we get a recurrence for time complexity as $T(n,k)=T(n, k-1)+ O(n^4)$ and $T(n,2)=O(n^4)$
which solves to $T(n,k)=O(kn^4)$. Note that $k\leq d$, the dimension of $\complex$. 
This gives an $O(d n^4)$ algorithm for computing the barcodes of persistent $k$-cup modules for all $k \in \{2,\dots,d \}$.

\subsection{Persistent cup-length: faster computation}
\label{sec:originalbetter}
The \emph{cup length} of a ring is defined as the maximum number of multiplicands that together give a nonzero product in the ring.
Let $\Int_{\ast} $ denote the set of all closed  intervals of $\mathbb{R}$.
Let $\mathcal{F}$ be an $\mathbb{R}$-indexed filtration of simplicial complexes. 
The \emph{persistent cup-length function} $\cupprod_{\bullet}:\Int_{\ast}\to\NBB$ 
is defined as a function 
from the set of closed intervals to the set of non-negative integers, which assigns to each interval $[ a,b]$, the cup-length of the image ring  $\im\big(\Homol^*(\complex)[ a,b]\big)$, which is the ring $\im\big(\Homol^*(\complex_b)\to\Homol^*(\complex_a)\big)$. 

Given a $P$-indexed filtration $\mathcal{F}$ of a $d$-complex $\complex$ of size $n$, let $V^k_\bullet$ denote
its persistent $k$-cup module.
Leveraging the fact that $\cupprod_{\bullet}([a,b])={\rm argmax}\{k \mid \rk_{V^k_\bullet}([a,b])\not = 0\}$ (see Proposition 5.9 in \cite{cuparxivtwo}),  the algorithm described in \Cref{sec:alg-kproduct} can be used to compute the persistent cup-length in $O(dn^4)$  time, 
whereas $O(n^{d+2})$ is a coarse estimate for the runtime of the algorithm described in \cite{contessoto}. 
Thus, for $d\geq 3$, our complexity bound for computing the persistent cup length is strictly better. 
We refer the reader to \Cref{sec:better} for further details.

\section{Partition modules of the cup product: a more refined invariant} \label{sec:partition}

A partition $\lambda_q$ of an integer $q$ is a multiset of integers that sum to $q$, written as $\lambda_q \vdash q$.
 That is, a multiset $\lambda_q = \{s_1,s_2,\dots,s_\ell\}$ is a partition of $q$ if $s_1 + s_2 + \dots + \dots s_\ell = q$. 
 The integers $s_1,s_2,\dots,s_\ell$ are  non-decreasing.
 For every partition $\lambda_q$ of $q$, we  define a submodule $ \im \Homol^{\lambda_q}(\smile \complex_\bullet))$ (written as $ \im \Homol^{\lambda_q}(\smile_\bullet))$ when $\complex$ is clear from context) of $ \im \Homol^{q}(\smile_\bullet^\ell))$:
 \[ \im \Homol^{\lambda_q}(\smile_i)) = \langle [\alpha_1] \smile [\alpha_2] \smile \dots \smile [\alpha_{\ell}] \mid [\alpha_j] \in \Homol^{s_j}(\complex_i) \text{ for } j\in[\ell]  \rangle. \]
The structure map $\im \Homol^{\lambda_q}(\smile_i)) \to \im \Homol^{\lambda_q}(\smile_{i-1}))$ is the restriction of $\cohommap_i$ to $\im \Homol^{\lambda_q}(\smile_i))$.

For an integer $q\geq 1$, let $\mathcal{P}(q)$ denote the number of partitions of $q$. In \cite{partsbound}, Pribitkin proved that for $q\geq 1$, $\mathcal{P}(q) < \frac{e^{c\sqrt{q}}}{q^{\frac{3}{4}}} $, where $c = \pi \sqrt{\nicefrac{2}{3}}$.
 For a $d$-complex $\complex$, 
let $\mathcal{P}^{\uparrow}(d)$ denote the total number of partition modules. Below, we obtain an  upper bound for $\mathcal{P}^{\uparrow}(d)$.

 \begin{equation*}
    \mathcal{P}^{\uparrow}(d) \,\, = \,\, \sum_{q=2}^d \mathcal{P}(q) \,\,< \,\, \sum_{q=2}^d \frac{e^{c\sqrt{q}}}{q^{\frac{3}{4}}} 
   \,\, < \,\, {d^{\frac{1}{4}}} e^{c\sqrt{d}}.
 \end{equation*}

When $d$ is small, as is often the case in practice, $ \mathcal{P}^{\uparrow}(d)$ is also small. For instance, $ \mathcal{P}^{\uparrow}(2) = 1$, $ \mathcal{P}^{\uparrow}(3) = 3$, $ \mathcal{P}^{\uparrow}(4) =  7$.

\paragraph*{Partition modules are more discriminative.} From \Cref{rem:diff} and \Cref{ex:nice}, it follows that barcodes of partition modules are a strictly finer invariant compared to barcodes of cup modules. 

\begin{remark} \label{rem:diff}
     Given two filtrations $\complex_{\bullet}$ and $\altcomplex_{\bullet}$, suppose that for some $\ell$ and $q$, $ \im \Homol^{q}(\smile^\ell \complex_\bullet))$ and $ \im \Homol^{q}(\smile^\ell \altcomplex_\bullet))$ are distinct. Without loss of generality, there exists a bar $(d,b]$ in $B( \im \Homol^{q}(\smile \complex_\bullet)))$ with no matching bar in $B( \im \Homol^{q}(\smile \altcomplex_\bullet)))$. Let $\zeta$ be a representative for the bar $(d,b]$. Then, $[\zeta]$ can be written as $[\zeta_1] \smile [\zeta_2] \smile \dots \smile [\zeta_{\ell}]$ in $\complex_b$. Let $s_i$ for each $i\in[\ell]$ denote the degree of cocycle class $[\zeta_i]$. Then, $\lambda_q = \{s_1,s_2,\dots,s_\ell\}$ is a partition of $q$. It follows that the bar $(d,b]$ will be present in  $ B( \im \Homol^{\lambda_q}(\smile \complex_\bullet)))$  but not in  $ B( \im \Homol^{\lambda_q}(\smile \altcomplex_\bullet)))$.
\end{remark}

\begin{example} \label{ex:nice}
    Let $\altcomplex^1  = (S^3 \times S^1) \vee S^2 \vee S^2$ and $\altcomplex^2  = (S^2 \times S^2) \vee S^1 \vee S^3$. The natural cell filtrations    $\altcomplex^1_\bullet$ and $\altcomplex^2_\bullet$ have isomorphic persistence modules and persistent cup modules. While  $\altcomplex^1_{\bullet}$ has a nontrivial barcode for $ \im \Homol^{(3,1)}$ and a trivial barcode for $ \im \Homol^{(2,2)}$, the opposite is true for $\altcomplex^2_{\bullet}$.

The barcodes for the persistence modules (using the convention from \Cref{sec:perscoh}) are
    \begin{align*}
       B(\Homol^0(\altcomplex^1_\bullet)) =B(\Homol^0(\altcomplex^2_\bullet)) &= \{(-1,4]\}, \\
       B(\Homol^1(\altcomplex^1_\bullet)) = B(\Homol^1(\altcomplex^2_\bullet)) &= \{(0,4]\}, \\
       B(\Homol^2(\altcomplex^1_\bullet)) = B(\Homol^2(\altcomplex^2_\bullet)) &= \{(1,4],(1,4]\}  \\
       B(\Homol^3(\altcomplex^1_\bullet)) = B(\Homol^3(\altcomplex^2_\bullet)) &= \{(2,4]\} \text{ and } \\
       B(\Homol^4(\altcomplex^1_\bullet)) = B(\Homol^4(\altcomplex^2_\bullet)) &= \{(3,4]\}. 
    \end{align*}
  For the persistent cup modules, $B(\im \Homol^{4}(\smile \altcomplex^1_{\bullet})) = B(\im \Homol^{4}(\smile \altcomplex^2_{\bullet})) = \{(3,4]\}$. For other degrees, the persistent cup modules are trivial.

  Finally, for partition modules $B(\im \Homol^{(2,2)}(\smile \altcomplex^2_{\bullet})) = \{(3,4]\}$ and $B(\im \Homol^{(2,2)}(\smile \altcomplex^1_{\bullet})) $ is empty, while   $B(\im \Homol^{(3,1)}(\smile \altcomplex^2_{\bullet}))$ is empty and $B(\im \Homol^{(3,1)}(\smile \altcomplex^1_{\bullet})) = \{(3,4]\}$.
  
\end{example}

From \Cref{ex:nicetwoexp}, we see that partition modules are not a complete invariant.

\begin{example} \label{ex:nicetwoexp}
    Let $\altaltcomplex^1$ be the  $3$-torus, and $\altaltcomplex^2  = \mathbb{RP}^2 \vee \mathbb{RP}^2 \vee \mathbb{RP}^3$. The natural cell filtrations $\altaltcomplex^1_\bullet$ and $\altaltcomplex^2_\bullet$ have isomorphic persistence modules, isomorphic persistent cup modules as well as isomorphic partition modules. Yet,  $\altaltcomplex^1$ and $\altaltcomplex^2$  have non-isomorphic cohomology algebras.

    The barcodes for the persistence modules are
    \begin{align*}
       B(\Homol^0(\altaltcomplex^1_\bullet)) =B(\Homol^0(\altaltcomplex^2_\bullet)) &= \{(-1,3]\}, \\
       B(\Homol^1(\altaltcomplex^1_\bullet)) = B(\Homol^1(\altaltcomplex^2_\bullet)) &= \{(0,3],(0,3],(0,3]\}, \\
       B(\Homol^2(\altaltcomplex^1_\bullet)) = B(\Homol^2(\altaltcomplex^2_\bullet)) &= \{(1,3],(1,3],(1,3]\} \text{ and }  \\
       B(\Homol^3(\altaltcomplex^1_\bullet)) = B(\Homol^3(\altaltcomplex^2_\bullet)) &= \{(2,3]\}. 
    \end{align*}   
     The barcodes for the persistence cup modules are
    \begin{align*}   
    B(\im \Homol^2(\smile \altaltcomplex^1_\bullet)) = B(\im \Homol^2(\smile \altaltcomplex^2_\bullet)) &= \{(1,3],(1,3],(1,3]\} \text{ and }  \\
       B(\im \Homol^3(\smile \altaltcomplex^1_\bullet)) = B(\im \Homol^3(\smile \altaltcomplex^2_\bullet)) &= \{(2,3]\}. 
    \end{align*}    
     The barcodes for the partition modules are
    \begin{alignat*}{3}   
    &B(\im \Homol^{(1,1)}(\smile \altaltcomplex^1_\bullet)) &&= B(\im \Homol^{(1,1)}(\smile \altaltcomplex^2_\bullet)) &&= \{(1,3],(1,3],(1,3]\},   \\
       &B(\im \Homol^{(2,1)}(\smile \altaltcomplex^1_\bullet)) &&= B(\im \Homol^{(2,1)}(\smile \altaltcomplex^2_\bullet)) &&= \{(2,3]\} \text{ and }   \\
       &B(\im \Homol^{(1,1,1)}(\smile \altaltcomplex^1_\bullet)) &&= B(\im \Homol^{(1,1,1)}(\smile \altaltcomplex^2_\bullet)) &&= \{(2,3]\}. 
   \end{alignat*}  
   The cohomology algebra $\Homol^{\ast}(\altaltcomplex^1) \approx \mathbb{Z}_2[a,b,c]/(a^2,b^2,c^2)$. 
   Note that $\Homol^{\ast}(\mathbb{RP}^2) \approx \mathbb{Z}_2[a]/(a^3)$ and $\Homol^{\ast}(\mathbb{RP}^3) \approx \mathbb{Z}_2[a]/(a^4)$.
   Let $\Homol^{>}$ denote the positive parts of $\Homol^{\ast}$. Then, the cohomology algebra of $\altaltcomplex^2$ is $\Homol^{\ast}(\altaltcomplex^2) \approx \mathbb{Z}_2\mathbf{1}\oplus \Homol^{>}(\mathbb{RP}^2) \oplus  \Homol^{>}(\mathbb{RP}^2) \oplus \Homol^{>}(\mathbb{RP}^3).$

   Unlike  $\Homol^{\ast}(\altaltcomplex^2)$, there does not exist a cocycle $x$ in the algebra $\Homol^{\ast}(\altaltcomplex^1)$ such that $x^3$ is nonzero. Hence $\Homol^{\ast}(\altaltcomplex^1)$ and $\Homol^{\ast}(\altaltcomplex^2)$ are non-isomorphic.
   
\end{example}

\bigskip

The barcodes of all the partition modules of the cup product can be computed in $ O({d^{\frac{1}{4}}} e^{c\sqrt{d}}n^4) $ time, where $c = \pi \sqrt{\nicefrac{2}{3}}$ time. The algorithm for computing them is described in \Cref{sec:algpart}. In \Cref{sec:stab}, using functoriality of the cup product, we observe that partition modules are stable for $\cechfull$ and $\ripsfull$ filtrations w.r.t. the interleaving distance.

\subsection{ Algorithm: partition modules} \label{sec:algpart}

To begin with,  {\sc CupPers2Parts} describes an algorithm for computing the barcode of the module $ \im \Homol^{\lambda_q}(\smile_\bullet))$ for $\lambda_q \vdash q$ when $|\lambda_q| = 2$. 
First, in Step $0$, we need to check if the barcode for the partition  $\lambda_q = \{s_1,s_2\}$ has already been computed because {\sc CupPers2Parts} is called from {\sc ExtendCupPersKParts}  possibly multiple times with the same argument $\lambda_q$. 
 In Step 1, we
compute the barcode of the cohomology persistence module $\mathsf{H}^\ast(\mathsf{K}_\bullet)$
along with a persistent cohomology basis.
As in {\sc CupPers2Parts}, a basis  is maintained with the matrix $\matrixh$  whose columns are (restricted) representative cocycles. The matrix $\matrixh$ is initialized with essential cocycles. The matrix $\matrixb$ is initialized with the coboundary matrix $\partial^{\perp}$ with standard cochain basis. Subsequently, nontrivial cocycle vectors are added to $\matrixb$. For every $k$, the classes of the nontrivial cocycles in matrix $\matrixb$ form a basis for $ \im \Homol^{\lambda_q}(\smile_k))$. In particular, a cocycle $\zeta = \xi_1 \cup \xi_2$ is added to $\basisB$ only if $\deg(\xi_1)=s_1$ and $\deg(\xi_2) = s_2$ or vice versa. Other than the details mentioned here, {\sc CupPers2Parts} is identical to {\sc CupPers}.

\medskip

\noindent
{\bf Algorithm} {\sc CupPers2Parts} ($\mathsf{K}_\bullet,\lambda_q$)\label{algorithm:cupperspart}
\begin{itemize}
    \item Step 0. If the barcode for the partition $\lambda_q = \{s_1,s_2\}$ has already been computed, then return the barcode with representatives 
    $\{(d_{i,2},b_{i,2}], \xi_{i,2}\}$. 
    \item Step 1. Compute barcode $\barcode =\{(d_i,b_i]\}$ of $\Homol^\ast(\complex_\bullet)$
    with representative cocycles $\xi_i$; Let $\matrixh=\{\xi_i  \mid [\xi_i] \text{ essential} \}$; Initialize $\matrixb$  with the coboundary matrix $\partial^{\perp}$ obtained by taking transpose of the boundary matrix $\partial$;
    \item Step 2. For $k:=n$ to $1$ do
    \begin{itemize}
        \item Restrict the cocycles in $\matrixb$ and $\matrixh$ to index $k$; 
        \item Step 2.1 For every $i$ s.t.  $k=b_i$ ($k$ is a birth-index)
        \begin{itemize}
            \item Step 2.1.1 If $k\not = n$, update $\matrixh:=[\matrixh \mid \xi_i]$
            \item Step 2.1.2 If $\deg(\xi_i) = s_1$ 
            \begin{enumerate}
            \item Step 2.1.2.1 For every $\xi_j \in \matrixh$ with $\deg(\xi_j) = s_2$ 
            \begin{enumerate}
                \item [i.] If $(\zeta\leftarrow \xi_i\smile \xi_j)\not = 0$ and $\zeta$ is
                independent in $\matrixb$, then $\matrixb:=[\matrixb \mid \zeta]$ with column $\zeta$ annotated
                as $\zeta\cdot {\rm birth}:=k$ and  $\zeta\cdot {\rm rep@birth}:= \zeta$
            \end{enumerate}
            \end{enumerate}
            \item Step 2.2.2 If $\deg(\xi_i) = s_2$ and $s_1 \neq s_2$ 
            \begin{enumerate}
            \item Step 2.2.2.1 For every $\xi_j \in \matrixh$ with $\deg(\xi_j) = s_1$ 
            \begin{enumerate}
                \item [i.] If $(\zeta\leftarrow \xi_i\smile \xi_j)\not = 0$ and $\zeta$ is
                independent in $\matrixb$, then $\matrixb:=[\matrixb \mid \zeta]$ with column $\zeta$ annotated
                as $\zeta\cdot {\rm birth}:=k$ and  $\zeta\cdot {\rm rep@birth}:= \zeta$
            \end{enumerate}
            \end{enumerate}
        \end{itemize}
        \item Step 2.2 If $k=d_i$ for some $i$ then ($k$ is a death-index)
        \begin{itemize}
         \item Step 2.2.1 Reduce $\matrixb$ with left-to-right column additions 
         \item Step 2.2.2 If a nontrivial cocycle $\zeta$ is zeroed out, remove $\zeta$ from $\matrixb$, generate the
         bar-representative pair $\{(k,\zeta\cdot {\rm birth}],\zeta\cdot {\rm rep@birth}\}$
         \item Step 2.2.3 Update $\matrixh$ by removing the column $\xi_i$
        \end{itemize}
    \end{itemize}
\end{itemize}

\begin{definition}[Refinement of a partition]
   Let $\lambda_q$ and $\lambda_q'$ be partitions of $q$. We say $\lambda_q$ \emph{refines} $\lambda_q'$ if the parts of $\lambda_q'$ can be subdivided to produce the parts of $\lambda_q$. 
\end{definition}

For example, $(1,1,1,1)  \vdash 4$ and $(1,2,1)  \vdash 4$  and  $(1,1,1,1)$ is a refinement of $(1,2,1)$.

\begin{remark}
 If a partition $\lambda_q$ is a refinement of a partition  $\lambda_q'$,  then   $\im \Homol^{\lambda_q}(\smile_\bullet))$ is a submodule of $\im \Homol^{\lambda_q'}(\smile_\bullet))$.
\end{remark}

\begin{definition}[Extension of a partition]
Let $p$ and $q$ be integers, with $q > p$.
   Let $\lambda_q = (s_1,s_2,\dots,s_m)$ be a partition of $q$ and $\lambda_p = (s'_1,s'_2,\dots,s'_\ell)$ be a partition of $p$  for some integers $\ell$ and $m$, with $m>\ell$. We say $\lambda_q$ \emph{extends} $\lambda_p$ if $s_i = s'_i$ for $i \in [\ell]$. We say that $\lambda_q$ \emph{extends} $\lambda_p$ \emph{by one} if $|\lambda_q| = |\lambda_p| + 1$.
\end{definition}

For example, $(2,2) \vdash 4$ and $(2,2,3) \vdash 5$, and $(2,2,3)$ \emph{extends $(2,2)$ by one}.  

\medskip

We now describe {\sc ExtendCupPersKParts} that gives an algorithm for computing the barcode of the module $ \im \Homol^{\lambda_t}(\smile_\bullet))$ for $\lambda_t \vdash t$.
In Step $0$, we check if the barcode for the partition  $\lambda_t$ has already been computed because {\sc ExtendCupPersKParts} is called recursively from {\sc ExtendCupPersKParts}  possibly multiple times with the same argument $\lambda_t$. 
In Step $1$, we first check if $|\lambda_t| = 2$, in which case, we invoke {\sc CupPers2Parts} and return.
Otherwise, $|\lambda_t| = k > 2$, and the algorithm calls itself recursively to generate the sets of bar-representative pairs for the module $ \im \Homol^{\lambda_q}(\smile_\bullet))$, where $\lambda_t$ is a partition that extends $\lambda_q$ by one. As in the case of {\sc OrderkCupPers}, the birth and death indices
of  order-$k$ product cocycle classes are subsets of birth and death indices resp.
of  ordinary persistence. Therefore, at each birth index of the cohomology module, we check if the cup product of a representative cocycle with degree $t-q$ (maintained
in matrix $\matrixh$) with a
representative for $ \im \Homol^{\lambda_q}(\smile_\bullet))$
(which has degree $q$ and is maintained in matrix $\matrixr$) generates 
a new cocycle in the barcode for $ \im \Homol^{\lambda_t}(\smile_\bullet))$  (Steps 2.1.2(i), 2.2.2(i)). If so,
we note this birth with the resp. cocycle (by annotating the column) and add it to
the matrix $\matrixb$ that maintains a basis for live order-$k$ product cocycles whose respective degrees form a partition $\lambda_t$ of $t$.The case of death (Step 2.3) is identical to {\sc OrderkCupPers}.

\medskip

\noindent
{\bf Algorithm} {\sc ExtendCupPersKParts} ($\mathsf{K}_\bullet$,$\lambda_t$)\label{algorithm:cupperskparts}
\begin{itemize}
    \item Step 0. If the barcode for the partition $\lambda_t$ has already been computed, then return the barcode with representatives 
    $\{(d_{i,k},b_{i,k}], \xi_{i,k}\}$. 
    Else, let  $\lambda_q$ be any partition such that $\lambda_t$ extends $\lambda_q$ by one, and let $k = |\lambda_t|$.  
    \item Step 1. If $|\lambda_t|=2$, return the barcode with representatives 
    $\{(d_{i,2},b_{i,2}], \xi_{i,2}\}$ 
    computed by {\sc CupPers2Parts}$(\mathsf{K}_\bullet,\lambda_t)$\\
    Set $\{(d_{i,k-1},b_{i,k-1}],\xi_{i,k-1}\}$ $\leftarrow$ {\sc ExtendCupPersKParts}($\mathsf{K}_\bullet$,$\lambda_q$)\\
  \hspace*{0.2in} Let $\matrixh=\{\xi_{i,1}  \mid [\xi_{i,1}] \text{ essential and }\deg(\xi_{i,1})=t-q \}$; 
  $\matrixr:=\{\xi_{i,k-1} \mid b_{i,k-1}=n\}$; $\matrixb:=\partial^{\perp}$;
    \item Step 2. For $\ell:=n$ to $1$ do
    \begin{itemize}
        \item Restrict the cocycles in $\matrixb$, $\matrixr$, and $\matrixh$ to index $\ell$; 
        \item Step 2.1 For every $r$ s.t.  $b_{r,1} = \ell \neq n$ (i.e., $\ell$ is a birth-index) and $\deg(\xi_{r,1})= t -q $  
        \begin{itemize}
            \item Step 2.1.1 Update $\matrixh:=[\matrixh \mid \xi_{r,1}]$ 
            \item Step 2.1.2 For every  $\xi_{j,k-1} \in \matrixr$ 
            \begin{enumerate}
                \item [i.] If $(\zeta\leftarrow \xi_{r,1}\smile \xi_{j,k-1})\not = 0$ and $\zeta$ is
                independent in $\matrixb$, then $\matrixb:=[\matrixb \mid \zeta]$ with column $\zeta$ annotated
                as $\zeta\cdot {\rm birth}:=\ell$ and $\zeta\cdot {\rm rep@birth}:= \zeta$
            \end{enumerate}
        \end{itemize}

        \item Step 2.2 For all $s$ such that  $\ell=b_{s,k-1}$      
        \begin{itemize}
            \item Step 2.2.1 If $\ell\not = n$, update  $\matrixr:=[\matrixr\mid \xi_{s,k-1}]$
            \item Step 2.2.2 For every  $\xi_{i,1} \in \matrixh$ 
            \begin{enumerate}
                \item [i.] If $(\zeta\leftarrow \xi_{s,k-1}\smile \xi_{i,1})\not = 0$ and $\zeta$ is
                independent in $\matrixb$, then $\matrixb:=[\matrixb \mid \zeta]$ with column $\zeta$ annotated
                as $\zeta\cdot {\rm birth}:=\ell$ and $\zeta\cdot {\rm rep@birth}:= \zeta$
            \end{enumerate}
        \end{itemize}        
        
        \item Step 2.3 If $\ell=d_{i,1}$ (i.e. $\ell$ is a death-index)  then 
        \begin{itemize}
         \item Step 2.3.1 Reduce $\matrixb$ with left-to-right column additions 
         \item Step 2.3.2 If a nontrivial cocycle $\zeta$ is zeroed out, remove $\zeta$ from $\matrixb$, generate the
         bar-representative pair $\{(\ell,\zeta\cdot {\rm birth}],\zeta\cdot {\rm rep@birth}\}$
         \item Step 2.3.3 Remove the column $\xi_{i,1}$ from $\matrixh$ 
         \item Step 2.3.4 Remove the column $\xi_{j,k-1}$ from $\matrixr$ if $d_{j,k-1}=\ell$ for some $j$
        \end{itemize}
    \end{itemize}
\end{itemize}

For every $k \in \{2,\dots,d\}$, {\bf Algorithm} {\sc ComputePartitionBarcodes} first generates all partitions of integer $k$, and then for every partition $\lambda_k$ of $k$ computes the barcode of the partition module $ \im \Homol^{\lambda_k}(\smile_\bullet))$.

\noindent
{\bf Algorithm} {\sc ComputePartitionBarcodes} ($\mathsf{K}_\bullet$)\label{algorithm:cmoputeparts}

\begin{itemize}
    \item Step 1. For $k:=2$ to $d$ do
    \begin{itemize}
        \item Step 1.1 Compute the set of partitions of $k$. Denote it by $\Lambda_k$.
        \item Step 1.2 For every partition $\lambda_k \in \Lambda_k$ do
        \begin{itemize}
             \item Step 1.2.1 $\{(d_{i,|\lambda_k}|,b_{i,|\lambda_k}|],\xi_{i,|\lambda_k}|\}$ $\leftarrow$ {\sc ExtendCupPersKParts}($\mathsf{K}_\bullet$,$\lambda_k$).
        \end{itemize}
    \end{itemize}
\end{itemize}

\paragraph*{Correctness and complexity.}
The correctness proofs for  {\sc CupPers} and {\sc ExtendCupPersKParts} are identical to those of {\sc CupPers2Parts} and {\sc OrderkCupPers}, respectively.

All partitions of an integer $k$ can be generated in output-sensitive time using partitions of integer $k-1$. For instance, see \cite{partitioncode} for a Python code to do the same. 
Hence, Step 1.1 of {\sc ComputePartitionBarcodes} runs in  time $O(\mathcal{P}^{\uparrow}(d))$ which is upper bounded by $
   O({d^{\frac{1}{4}}} e^{c\sqrt{d}})$, where $c = \pi \sqrt{\nicefrac{2}{3}}$  (See \Cref{sec:partition}). 
Note that {\sc ExtendCupPersKParts} (and  {\sc CupPers2Parts}) executes beyond Steps $0$ with a parameter $\lambda_k$ only when it is called for the first time with that parameter. The total number of calls to {\sc ExtendCupPersKParts} that proceed to Steps $1$ is, therefore, bounded by $\mathcal{P}^{\uparrow}(d)$.
If there are subsequent recursive calls to {\sc ExtendCupPersKParts} with $\lambda_k$ as a parameter it returns at Step $0$. Note that {\sc ExtendCupPersKParts} calls itself recursively only once (in Step $1$). So the total number of calls where {\sc ExtendCupPersKParts} returns at Step $0$ is bounded by $\mathcal{P}^{\uparrow}(d)$. If {\sc ExtendCupPersKParts} returns at Step $0$, the cost of execution is $O(1)$, else it is $O(n^4)$. Hence, the total cost of   Step 1.2 of  {\sc ComputePartitionBarcodes} is $\mathcal{P}^{\uparrow}(d) O(n^4)$ which is $O({d^{\frac{1}{4}}} e^{c\sqrt{d}} n^4)$.

\subsection{ Stability} \label{sec:stab}

We  establish stability of partition  modules of the cup product for  Rips and $\cechfull$ complexes. In particular, we show that when the Gromov-Haudorff distance  (Hausdorff distance) between a point cloud and its perturbation is bounded by a small constant, then the interleaving distance between barcodes of respective Rips ($\cechfull$)partition modules is also bounded by a small constant.
\subsubsection{Geometric complexes}
\begin{definition}[ Rips complexes]
Let $X$ be a finite point set in $\mathbb{R}^d$. 
The  Rips complex of $X$ at scale $t$ consists of all simplices with diameter at most $t$, where the diameter of a simplex is the maximum distance between any two points in the simplex. In other words, 
\[ \Rips_t(X) = \{ S \subset X \mid \diam S \leq t  \}. \]
The Rips filtration of $X$, denoted by $\Rips_{\bullet}(X)$, is the nested sequence of complexes $\{\Rips_t(X)\}_{t\geq 0}$, where $\Rips_s(X) \subseteq \Rips_t(X) $ for $s \leq t$. 
\end{definition}

\begin{definition}[$\cechfull$ complexes]
Let $X$ be a finite point set in $\mathbb{R}^d$. Let $D_{r,x}$ denote a Euclidean ball of radius $r$ centered at $x$. The $\cechfull$ complex of $X$ for radius $r$ 
consists of all simplices satisfying the following condition:
\[
\Cech_r(X) = \{ S \subset X \mid \bigcap_{x\in S}  D_{r,x}\not = \emptyset\}.
\]

The $\cechfull$ filtration of $X$, denoted by $\Cech_{\bullet}(X)$, is the nested sequence of complexes $\{\Cech_r(X)\}_{r\geq 0}$, where $\Cech_s(X) \subseteq \Cech_t(X) $ for $s \leq t$. 
\end{definition}

\subsubsection{Gromov-Hausdorff distance}

Let $X$ and $Y$ be compact subspaces of a metric space $M$ with
distance $d$.
For a point $p\in X$, $d(p;Y)$ is defined as
\[
d(p,Y)=\inf\left\{ d(p,q)\mid q\in Y\right\} 
\]
and the distance $d(X,Y)$ between spaces $X$ and $Y$ is defined as 
\[
d(X,Y)=\sup\left\{ d(p,Y)\mid p\in X\right\}. 
\]
The Hausdorff distance $d_H$ between $X$ and $Y$ is defined as
\[
d_{H}(X,Y)=\max\left\{ d(X,Y),d(Y,X)\right\}. 
\]
The Gromov-Hausdorff distance $d_{GH}$ between $X$ and Y is defined as 
\[
d_{GH}(X,Y)=\inf\left\{ d_{H}(f(X);g(Y))\mid f:X\hookrightarrow M\,\,,\,\,g:Y\hookrightarrow M\right\} 
\]
where the infimum is taken over all isometric embeddings $f:X\hookrightarrow M\,\,,\,\,g:Y\hookrightarrow M$
into some common metric space $M$.

\subsubsection{Stability of partition modules} \label{sec:stability}

In this section, as a direct consequence of the functoriality of the cup product, we show that  the partition modules are stable for $\cechfull$ and $\ripsfull$ filtrations.

To begin with, let $d_{I}(M,N)$ denote the interleaving distance between two persistence 
modules $M$ and $N$~\cite{chazal2014persistence}. For finite point sets $X$ and $Y$ in $\mathbb{R}^d$, let $d_{H}(X,Y)$ denote the Hausdorff distance, and let $d_{GH}(X,Y)$ denote the Gromov-Hausdorff distance  between them.
Let $\Rips_{\bullet}(X)$ and $\Rips_{\bullet}(Y)$ denote the respective Rips filtrations of $X$ and $Y$, and let $\Cech_\bullet(X)$ and $\Cech_\bullet(Y)$ denote the respective $\cechfull$ filtrations of $X$ and $Y$.

\begin{theorem} \label{thm:ripsstable}
Let $\lambda_q = \{s_1,s_2,\dots,s_\ell\}$ be a partition of an integer $q$.
Then, for finite point sets $X$ and $Y$ in $\mathbb{R}^d$, the following identities hold true:
\begin{align*}
    \frac{1}{2}d_{\inter}(\im \Homol^{\lambda_q}(\smile \Rips_{\bullet}(X)),\im \Homol^{\lambda_q}(\smile \Rips_{\bullet}(Y))) &\leq d_{GH}(X,Y). \\ 
    \frac{1}{2}d_{\inter}(\im \Homol^{\lambda_q}(\smile \Cech_{\bullet}(X)),\im \Homol^{\lambda_q}(\smile \Cech_{\bullet}(Y))) &\leq  d_{H}(X,Y).
\end{align*}
\end{theorem}

\begin{proof}
    Let $X$ and $Y$ be  point sets in a common Euclidean space $\mathbb{R}^{d}$ such that
$d_{GH}(X,Y) = \frac{\epsilon}{2}$.
Then, in the proof of Lemma 4.3 of \cite{chazal2014persistence}, Chazal et al.  showed that $\Rips_{\bullet}(X)$ and $\Rips_{\bullet}(Y)$ are $\epsilon$-interleaved.  

\medskip

 \adjustbox{scale=0.75}{
\begin{tikzcd}
	\ldots & {\Rips_{a}(X)} && {\Rips_{a+\epsilon}(X)} && {\Rips_{a+2\epsilon}(X)} && \ldots \\
	\\
	\ldots & {\Rips_{a}(Y)} && {\Rips_{a+\epsilon}(Y)} && {\Rips_{a+2\epsilon}(Y)} && \ldots
	\arrow[from=3-2, to=1-4]
	\arrow[from=1-2, to=3-4]
	\arrow[from=1-4, to=3-6]
	\arrow[from=3-4, to=1-6]
	\arrow[from=3-6, to=1-8]
	\arrow[from=1-6, to=3-8]
	\arrow[hook, from=1-2, to=1-4]
	\arrow[hook, from=1-4, to=1-6]
	\arrow[hook, from=3-2, to=3-4]
	\arrow[hook, from=3-4, to=3-6]
	\arrow[from=1-1, to=1-2]
	\arrow[from=3-1, to=3-2]
	\arrow[from=3-6, to=3-8]
	\arrow[from=1-6, to=1-8]
\end{tikzcd}
}

\medskip

Applying the cohomology functor, we obtain an $\epsilon$-interleaving of the respective cohomology persistence modules.
Let $\{\cohommap_{a',a}\}_{a',a\in \mathbb{R}}$ and $\{\cohommaptwo_{a',a}\}_{a',a\in \mathbb{R}}$ denote the structure maps for the modules $\Homol^{\ast}(\Rips_{\bullet}(X))$ and $\Homol^{\ast}(\Rips_{\bullet}(Y))$, respectively.  Also, let $F_{a+\epsilon}: \Homol^{\ast}(\Rips_{a+\epsilon}(X)) \to \Homol^{\ast}(\Rips_{a}(Y))$ and $G_{a+\epsilon}: \Homol^{\ast}(\Rips_{a+\epsilon}(Y)) \to \Homol^{\ast}(\Rips_{a}(X))$ for all $a\in \mathbb{R}$ be the maps that assemble to give an $\epsilon$-interleaving between $\Homol^{\ast}(\Rips_{\bullet}(X))$ and $\Homol^{\ast}(\Rips_{\bullet}(Y))$.

\medskip
 \adjustbox{scale=0.75}{
\begin{tikzcd}
	\ldots & { \Homol^{\ast}( \Rips_{a}(X))} && { \Homol^{\ast}( \Rips_{a+\epsilon}(X))} && { \Homol^{\ast}( \Rips_{a+2\epsilon}(X))} && \ldots \\
	\\
	\ldots & {\Homol^{\ast}( \Rips_{a}(Y))} && {\Homol^{\ast}( \Rips_{a+\epsilon}(Y))} && { \Homol^{\ast}( \Rips_{a+2\epsilon}(Y))} && \ldots
	\arrow["{F_{a+\epsilon}}"{pos=0.7}, tail reversed, no head, from=3-2, to=1-4]
	\arrow["{G_{a+\epsilon}}"'{pos=0.7}, tail reversed, no head, from=1-2, to=3-4]
	\arrow["{F_{a+2\epsilon}}"{pos=0.7}, tail reversed, no head, from=3-4, to=1-6]
	\arrow["{F_{a+3\epsilon}}"{pos=0.7}, tail reversed, no head, from=3-6, to=1-8]
	\arrow["{G_{a+3\epsilon}}"'{pos=0.7}, tail reversed, no head, from=1-6, to=3-8]
	\arrow["{\cohommap_{a+\epsilon,a}}", tail reversed, no head, from=1-2, to=1-4]
	\arrow["{\cohommap_{a+2\epsilon,a+\epsilon}}", tail reversed, no head, from=1-4, to=1-6]
	\arrow["{\cohommaptwo_{a+\epsilon,a}}"', tail reversed, no head, from=3-2, to=3-4]
	\arrow["{\cohommaptwo_{a+2\epsilon,a+\epsilon}}"', tail reversed, no head, from=3-4, to=3-6]
	\arrow[tail reversed, no head, from=1-1, to=1-2]
	\arrow[tail reversed, no head, from=3-1, to=3-2]
	\arrow[tail reversed, no head, from=3-6, to=3-8]
	\arrow[tail reversed, no head, from=1-6, to=1-8]
	\arrow["{G_{a+2\epsilon}}"'{pos=0.7}, tail reversed, no head, from=1-4, to=3-6]
\end{tikzcd}
}
\medskip

For every  $j\in[\ell]$, let $[\alpha_j] \in \Homol^{s_j}(\complex_i)$. Then, by the functoriality of the cup product,
$\cohommap_{a+\epsilon,a}([\alpha_1]\smile [\alpha_2]\smile \dots \smile[\alpha_\ell] ) = \cohommap_{a+\epsilon,a}([\alpha_1])\smile \cohommap_{a+\epsilon,a}([\alpha_2]) \smile \dots \smile \cohommap_{a+\epsilon,a}([\alpha_\ell])$, and hence for all $a\in \mathbb{R}$, $\cohommap_{a+\epsilon,a}$
restricts to a map $\im \Homol^{\lambda_q}(\smile \Rips_{a+\epsilon}(X)) \to \im \Homol^{\lambda_q}(\smile \Rips_{a}(X))$.

The functoriality of the cup product also gives us the restrictions $\cohommaptwo_{a+\epsilon,a}: \im \Homol^{\lambda_q}(\smile \Rips_{a+\epsilon}(Y)) \to \im \Homol^{\lambda_q}(\smile \Rips_{a}(Y))$, $F_{a+\epsilon}: \im \Homol^{\lambda_q}(\smile \Rips_{a+\epsilon}(X)) \to \im \Homol^{\lambda_q}(\smile \Rips_{a}(Y))$ and $G_{a+\epsilon}: \im \Homol^{\lambda_q}(\smile \Rips_{a+\epsilon}(Y)) \to \im \Homol^{\lambda_q}(\smile \Rips_{a}(X))$.
It is easy to check that the 
 restrictions of the maps  $\{F_{a+\epsilon}\}_{a\in \mathbb{R}}$ and  $\{G_{a+\epsilon}\}_{a\in \mathbb{R}}$ assemble to give an 
$\epsilon$-interleaving between the persistence modules $\im \Homol^{\lambda_q}(\smile \Rips_{\bullet}(X))$ and $\im \Homol^{\lambda_q}(\smile \Rips_{\bullet}(Y))$ with the restrictions of ${\{\cohommap_{a,a'}\}}_{a,a'\in \mathbb{R}}$ and ${\{\cohommaptwo_{a,a'}\}}_{a,a'\in \mathbb{R}}$ as the structure maps for $\im \Homol^{\lambda_q}(\smile \Rips_{\bullet}(X))$ and $\im \Homol^{\lambda_q}(\smile \Rips_{\bullet}(Y))$, respectively.

\medskip
 \adjustbox{scale=0.75}{
\begin{tikzcd}
\ldots & {\im \Homol^{\lambda_q}(\smile \Rips_{a}(X))} && {\im \Homol^{\lambda_q}(\smile \Rips_{a+\epsilon}(X))} && {\im \Homol^{\lambda_q}(\smile \Rips_{a+2\epsilon}(X))} && \ldots \\
	\\
	\ldots & {\im \Homol^{\lambda_q}(\smile \Rips_{a}(Y))} && {\im \Homol^{\lambda_q}(\smile \Rips_{a+\epsilon}(Y))} && {\im \Homol^{\lambda_q}(\smile \Rips_{a+2\epsilon}(Y))} && \ldots
	\arrow["{F_{a+\epsilon}}"{pos=0.7}, tail reversed, no head, from=3-2, to=1-4]
	\arrow["{G_{a+\epsilon}}"'{pos=0.7}, tail reversed, no head, from=1-2, to=3-4]
	\arrow["{F_{a+2\epsilon}}"{pos=0.7}, tail reversed, no head, from=3-4, to=1-6]
	\arrow["{F_{a+3\epsilon}}"{pos=0.7}, tail reversed, no head, from=3-6, to=1-8]
	\arrow["{G_{a+3\epsilon}}"'{pos=0.7}, tail reversed, no head, from=1-6, to=3-8]
	\arrow["{\cohommap_{a+\epsilon,a}}", tail reversed, no head, from=1-2, to=1-4]
	\arrow["{\cohommap_{a+2\epsilon,a+\epsilon}}", tail reversed, no head, from=1-4, to=1-6]
	\arrow["{\cohommaptwo_{a+\epsilon,a}}"', tail reversed, no head, from=3-2, to=3-4]
	\arrow["{\cohommaptwo_{a+2\epsilon,a+\epsilon}}"', tail reversed, no head, from=3-4, to=3-6]
	\arrow[tail reversed, no head, from=1-1, to=1-2]
	\arrow[tail reversed, no head, from=3-1, to=3-2]
	\arrow[tail reversed, no head, from=3-6, to=3-8]
	\arrow[tail reversed, no head, from=1-6, to=1-8]
	\arrow["{G_{a+2\epsilon}}"'{pos=0.7}, tail reversed, no head, from=1-4, to=3-6]
\end{tikzcd}}

The above diagram, proves the first claim.

Cohen-Steiner et al.~\cite{CEH07} showed that if $d_H(X,Y) = \frac{\epsilon}{2}$, then there exists an $\epsilon$-interleaving between $\Cech_{\bullet}(X)$ and $\Cech_{\bullet}(Y)$. Using this fact and repeating the argument above, we obtain the following the second claim.
\end{proof}

Thus, if the Gromov-Hausdorff distance between point sets $X$ and $Y$ is small, then the interleaving distance for the respective ordinary persistence modules, cup modules and partition modules of cup product are all small. 

\section{Relative cup modules} \label{sec:relcup}

Let $(\complex,\altcomplex)$ be a simplical pair.
As in the case of absolute cohomology, for the relative cup product, we have bilinear maps    \[\smile:\Chain^p(\complex,\altcomplex)\times\Chain^q(\complex,\altcomplex)\to\Chain^{p+q}(\complex,\altcomplex)\text{ \,\, that assemble to give a linear map \,\,}\] \[
     \smile:\Chain^*(\complex,\altcomplex)\otimes\Chain^*(\complex,\altcomplex)\to\Chain^{*}(\complex,\altcomplex).
\] 
Also, we have bilinear maps  
 \[\smile:\Homol^p(\complex,\altcomplex)\times\Homol^q(\complex,\altcomplex)\to\Homol^{p+q}(\complex,\altcomplex) \text{ \,\, that assemble to give a linear map \,\,}\]
 \[    \smile:\Homol^*(\complex,\altcomplex)\otimes\Homol^*(\complex,\altcomplex)\to\Homol^{*}(\complex,\altcomplex).
\] 

For a filtered complex $\complex$, its persistent relative cohomology is given by $\Homol^{\ast}(\complex,\complex_\bullet)$ with linear maps given by inclusions~\cite{de2011dualities}. Written in our convention for intervals, every finite bar $(d,b]$ in  $B(\Homol^{i}(\complex_\bullet))$, we have a corresponding finite bar $(d,b]$ in $B(\Homol^{i+1}(\complex,\complex_\bullet))$, and for every infinite bar $(d,n]$ in $B(\Homol^{i}(\complex_\bullet))$, we have an infinite bar $(-1,d]$ in $B(\Homol^{i}(\complex,\complex_\bullet))$.

\paragraph*{Defining relative cup modules.}

Consider the following homomorphism given by cup products:
 \begin{equation} \label{eqn:tensormaptworel}
     \smile_{\bullet}:\Chain^*(\complex,\complex_{\bullet})\otimes\Chain^*(\complex,\complex_{\bullet})\to\Chain^{*}(\complex,\complex_{\bullet}).
 \end{equation}
Taking $G_{\bullet}=\smile_{\bullet}$ in the definition of
image persistence, we get
a persistence module, denoted by $\im \rel \Homol^{*}(\smile \complex_{\bullet})$, which is called the \emph{persistent relative  cup module}. Whenever the underlying filtered complex is clear from the context, we  use the shorthand notation   $\im \rel \Homol^{*}(\cupmap)$ instead of $\im \Homol^{*}(\smile \complex_{\bullet})$. 

 \paragraph*{Defining relative $k$-cup modules.}
 Consider  image persistence of the map
 \begin{equation} \label{eqn:tensormapfourrel}
     \cupmap^{k}:\Chain^{*}(\complex,\complex_{\bullet})\otimes\Chain^{*}(\complex,\complex_{\bullet}) \otimes \dots \otimes\Chain^{*}(\complex,\complex_{\bullet}) \to\Chain^{*}(\complex,\complex_{\bullet})
 \end{equation}
  where the tensor product is taken $k$ times. 
 Taking $G_{\bullet} \,\,= \,\,\cupmap^k$ in the definition of
image persistence, we get
the \emph{persistent relative $k$-cup module module} $\im \rel \Homol^{*}(\cupmap^k)$.

\subsection{Lack of duality}

In contrast to ordinary persistence, the following examples highlight the fact that the barcodes of persistent (absolute) cup modules differ from barcodes of persistent relative cup modules. In fact, through \Cref{ex:flip,ex:flop}, we  observe that there does not exist a bijection between corresponding intervals. 

\begin{example} \label{ex:flip}

Let $\complex$ be a torus with a disk removed.  A torus can be obtained by identifying the opposite sides of a $[-1,1]^2$ square. The space $\complex$ can be obtained by removing a circle of radius $1$ around the origin. We now give the following CW structure to $\complex$: Let $x_0$ and $x_1$ be the $0$-cells, $p$, $q$, $r$ and $s$ be the $1$-cells and $\alpha$ be the $2$-cell. $p$ and $q$ are loops around $x_0$, $r$ joins $x_0$ and $x_1$, and $s$ is a loop around $x_1$. The attachment of the $2$-cell $\alpha$ is given by the word $p q p^{-1} q^{-1} r s r ^{-1}$. See  \Cref{fig:torus-disk-rem} for an illustration. 

\begin{figure}[htbp]
\centerline{\includegraphics[scale=0.7]{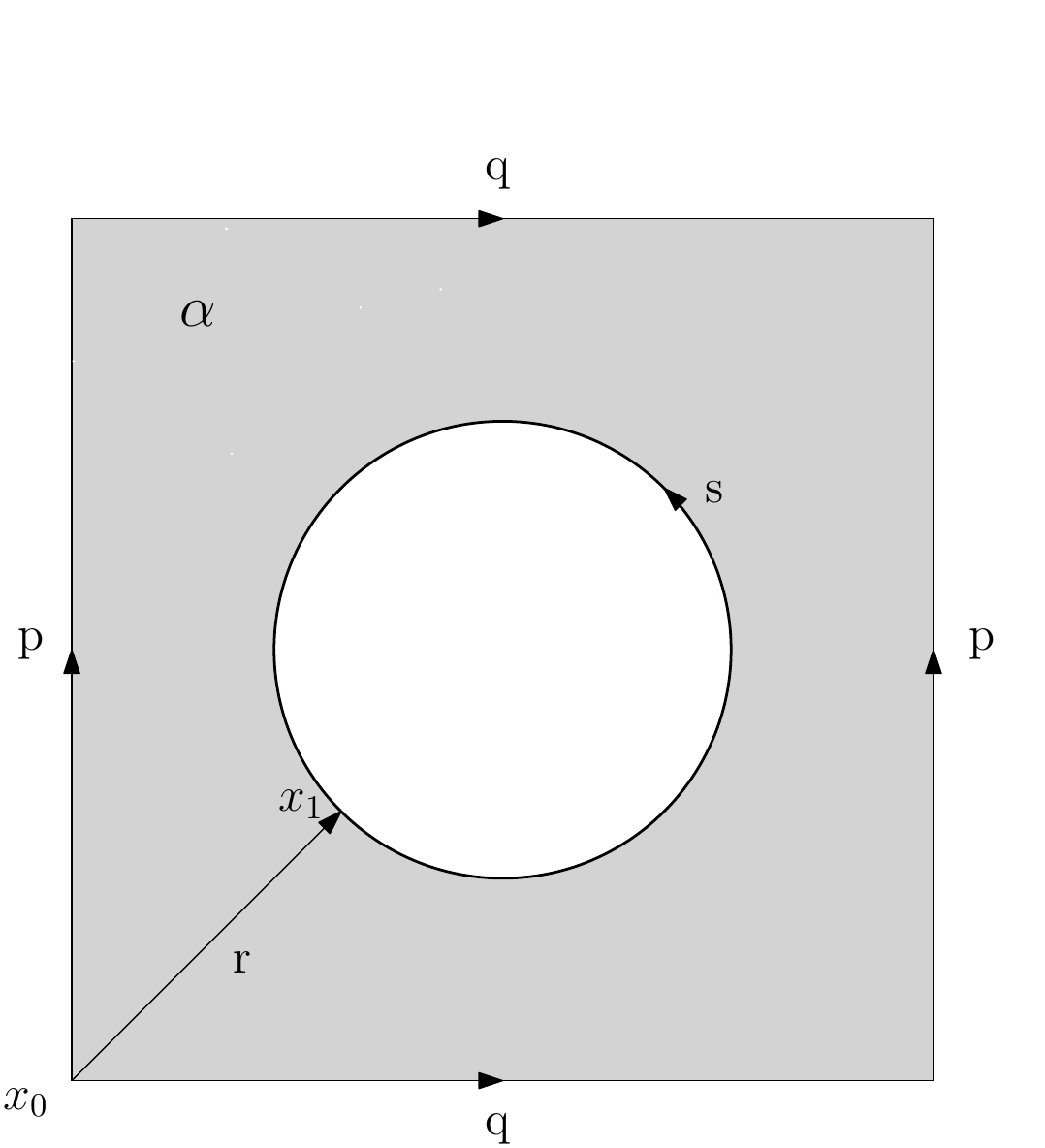}}
\caption{Complex $\complex$ is a torus with a disk removed.}
\label{fig:torus-disk-rem}
\end{figure}

Consider the cellular filtration $\complex_\bullet$ on $\complex$:
\begin{align*}
\complex_{0} & =\left\{ x_{0},x_{1}\right\}, \\
\complex_{1} & =\complex_{0}\cup\left\{ s,r\right\}, \\
\complex_{2} & =\complex_{1}\cup\left\{ p,q\right\}, \\
\complex_{3} & =\complex_{2}\cup\left\{ \alpha\right\}. 
\end{align*}

It is easy to check that the persistent (absolute) cup module for $\complex_{\bullet}$ is trivial. However, since $\complex_3 / \complex_1$ is a torus, the persistent relative cup module is nontrivial.

\end{example}

\begin{example} \label{ex:flop}

Let $\altcomplex'$ be a torus realized as a CW complex with a $0$-cell $x$, two $1$-cells $a$ and $b$ and a $2$-cell $\beta$. We now add a $2$-cell $\alpha$ to $\altcomplex'$ to obtain a CW complex $\altcomplex = \altcomplex' \cup \{\alpha\}$. See \Cref{fig:torus-with-disk} for an illustration. 

\todo{Figure spacing}

\begin{figure}[htbp]
\centerline{\includegraphics[scale=1.5]{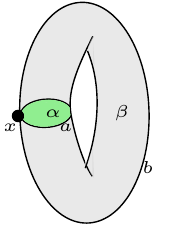}}
\caption{Complex $\altcomplex$ is a torus with a disk added.}
\label{fig:torus-with-disk}
\end{figure}

Now consider the following cellular filtration $\altcomplex_\bullet$ on $\altcomplex$: 
\begin{align*}
\altcomplex_{0} & =\left\{ x\right\} \\
\altcomplex_{1} & =\altcomplex_{0}\bigcup\left\{ a,b\right\} \\
\altcomplex_{2} & =\altcomplex_{1}\bigcup\left\{ \beta\right\} \\
\altcomplex_{3} & =\altcomplex_{2}\bigcup\left\{ \alpha\right\} 
\end{align*}

For the filtration $\altcomplex_\bullet$, the persistent (absolute) cup module is nontrivial since $\altcomplex_2$ is a torus. On the other hand, it is easy to check that the persistent relative cup module is trivial.
\end{example}

    \subsection{Algorithm: relative cup modules} \label{sec:algrel}

   In \Cref{sec:alg,sec:alg-kproduct}, we devised algorithms to compute (absolute) persistent $k$-cup modules.  The algorithms for computing \emph{relative} persistent $k$-cup modules are minor variations.
First, we describe how to compute the barcode of 
$\productrelmodulek$, which being an image module is a submodule
of  $\mathsf{H}^\ast(\complex,\mathsf{K}_\bullet)$. The vector space $\im \rel \mathsf{H}^\ast(\cupmapi)$ is a subspace of the  vector space 
$\mathsf{H}^\ast(\complex, \mathsf{K}_i)$. Let us call this subspace  
the \emph{relative cup space} of $\mathsf{H}^\ast(\complex, \mathsf{K}_i)$.
 {\sc RelCupPers} describes this algorithm to compute relative cup modules. First, in Step 0, we
compute the barcode of the cohomology persistence module $\mathsf{H}^\ast(\complex, \mathsf{K}_\bullet)$
along with a relative persistent cohomology basis.
This can be achieved in $O(n^3)$ time by applying the standard algorithm on the anti-transpose of the boundary matrix~\cite[Section 3.4]{de2011dualities}.
The basis $H$ is maintained with the matrix $\matrixh$  whose columns are representative cocycles. The matrix $\matrixh$ is initialized with the empty matrix. $\partial^{\perp}$ maintains the relative coboundaries as one processes the matrix in the reverse filtration order. At index $n$,  $\partial^{\perp}$ is empty.  Throughout, $ \partial^{\perp}$ is stored in the leftmost $n$ columns of $\matrixb$, and there are no other columns in $\matrixb$ at index $n$. Subsequently, nontrivial relative cocycle vectors are added to $\matrixb$. The classes of the nontrivial cocycles in matrix $\matrixb$ form a basis $\basisB$ for the relative cup space at any point in the course of the algorithm. 
In Step 2, at each index $k$, the $k$-th column of $\partial^{\perp}$ is populated with the coboundary of $k$. The remainder of the birth case and the whole of the death case is handled exactly like {\sc RelCupPers}. The correctness and complexity proofs for {\sc RelCupPers} are identical to  {\sc CupPers}.

\noindent
{\bf Algorithm} {\sc RelCupPers} ($\mathsf{K}_\bullet$)\label{algorithm:relcuppers}
\begin{itemize}
    \item Step 0. Compute barcode $\barcode =\{(d_i,b_i]\}$ of $\Homol^\ast(\complex,\complex_\bullet)$
    with representative cocycles $\xi_i$  
  \item Step 1. Initialize an $n \times n$ coboundary matrix $\partial^{\perp}$ as the zero matrix; $\partial^{\perp}$ is maintained as a submatrix of $\matrixb$; Initially all columns in  $\matrixb$ come from columns in $\partial^{\perp}$. Subsequently, in the course of the algorithm, new columns are added  to (and removed from) the right of $\partial^{\perp}$ in $\matrixb$ and the entries of $\partial^{\perp}$ are also modified; Initialize $\matrixh$ with the empty matrix
    \item Step 2. For $k:=n$ to $1$ do
    \begin{itemize}
        \item For every simplex $\sigma_j$ that has $\sigma_k$ as a face, set $\partial^{\perp}_{j,k}=1$
        \item Step 2.1 For every $i$ with  $k=b_i$ ($k$ is a birth-index)  and $\deg(\xi_i)>0 $ 
        \begin{itemize}
            \item Step 2.1.1 Update $\matrixh:=[\matrixh \mid \xi_i]$
            \item Step 2.1.2 For every $\xi_j \in \matrixh$ 
            \begin{enumerate}
                \item [i.] If $(\zeta\leftarrow \xi_i\smile \xi_j)\not = 0$ and $\zeta$ is
                independent in $\matrixb$, then $\matrixb:=[\matrixb \mid \zeta]$ with column $\zeta$ annotated
                as $\zeta\cdot {\rm birth}:=k$ and  $\zeta\cdot {\rm rep@birth}:= \zeta$
            \end{enumerate}
        \end{itemize}
        \item Step 2.2 If $k=d_i$ ($k$ is a death-index) for some $i$ and $\deg(\xi_i)>0 $ then 
        \begin{itemize}
         \item Step 2.2.1 Reduce $\matrixb$ with left-to-right column additions 
         \item Step 2.2.2 If a nontrivial cocycle $\zeta$ is zeroed out, remove $\zeta$ from $\matrixb$, generate the
         bar-representative pair $\{(k,\zeta\cdot {\rm birth}],\zeta\cdot {\rm rep@birth}\}$
         \item Step 2.2.3 Update $\matrixh$ by removing the column $\xi_i$        
    \end{itemize}
\end{itemize}

\end{itemize}

In  Algorithm {\sc RelOrderkCupPers},  The initialization and maintenance of the matrix  $\matrixb$ and $\partial^\perp$ is the same as  for  {\sc RelCupPers}. The matrices $\matrixh$ and $\matrixr$ are intialized with empty matrices. The remainder of the birth case and the whole of the death case are identical to {\sc OrderkCupPers}. The correctness and complexity proofs for {\sc RelOrderkCupPers} are identical to  {\sc OrderkCupPers}.

\medskip

\noindent
{\bf Algorithm} {\sc RelOrderkCupPers} ($\mathsf{K}_\bullet$,$k$)\label{algorithm:relorderkcuppers}
\begin{itemize}
    \item Step 0. If $k=2$, return the barcode with representatives 
    $\{(d_{i,2},b_{i,2}], \xi_{i,2}\}$ 
    computed by {\sc CupPers} on
    $\mathsf{K}_\bullet$,
    else $\{(d_{i,k-1},b_{i,k-1}],\xi_{i,k-1}\}$ $\leftarrow$ {\sc RelOrderkCupPers}($\mathsf{K}_\bullet$, $k-1$)
    \item Step 1. Initialize an $n \times n$ coboundary matrix $\partial^{\perp}$ as the zero matrix; $\partial^{\perp}$ is maintained as a submatrix of $\matrixb$; Initially all columns in  $\matrixb$ come from columns in $\partial^{\perp}$. Subsequently, in the course of the algorithm, new columns are added  to (and removed from) the right of $\partial^{\perp}$ in $\matrixb$ and the entries of $\partial^{\perp}$ are also modified; Initialize $\matrixh$ and $\matrixr$ with empty matrices 
     \item Step 2. For $\ell:=n$ to $1$ do
    \begin{itemize}
       \item For every simplex $\sigma_j$ that has $\sigma_k$ as a face, set $\partial^{\perp}_{j,k}=1$
        \item Step 2.1 For  every $r$ s.t. $b_{r,1} = \ell \neq n$ (i.e., $\ell$ is a birth-index) and $\deg(\xi_{r,1})>0 $  
        \begin{itemize}
            \item Step 2.1.1 Update $\matrixh:=[\matrixh \mid \xi_{r,1}]$ 
            \item Step 2.1.2 For every  $\xi_{j,k-1} \in \matrixr$ 
            \begin{enumerate}
                \item [i.] If $(\zeta\leftarrow \xi_{r,1}\smile \xi_{j,k-1})\not = 0$ and $\zeta$ is
                independent in $\matrixb$, then $\matrixb:=[\matrixb \mid \zeta]$ with column $\zeta$ annotated
                as $\zeta\cdot {\rm birth}:=\ell$ and $\zeta\cdot {\rm rep@birth}:= \zeta$
            \end{enumerate}
        \end{itemize}

        \item Step 2.2 For all $s$ such that  $\ell=b_{s,k-1}$      
        \begin{itemize}
            \item Step 2.2.1 If $\ell\not = n$, update  $\matrixr:=[\matrixr\mid \xi_{s,k-1}]$
            \item Step 2.2.2 For every  $\xi_{i,1} \in \matrixh$ 
            \begin{enumerate}
                \item [i.] If $(\zeta\leftarrow \xi_{s,k-1}\smile \xi_{i,1})\not = 0$ and $\zeta$ is
                independent in $\matrixb$, then $\matrixb:=[\matrixb \mid \zeta]$ with column $\zeta$ annotated
                as $\zeta\cdot {\rm birth}:=\ell$ and $\zeta\cdot {\rm rep@birth}:= \zeta$
            \end{enumerate}
        \end{itemize}        
        
        \item Step 2.3 If $\ell=d_{i,1}$ (i.e. $\ell$ is a death-index)  and $\deg(\xi_{i,1})>0 $  for some $i$ then 
        \begin{itemize}
         \item Step 2.3.1 Reduce $\matrixb$ with left-to-right column additions 
         \item Step 2.3.2 If a nontrivial cocycle $\zeta$ is zeroed out, remove $\zeta$ from $\matrixb$, generate the
         bar-representative pair $\{(\ell,\zeta\cdot {\rm birth}],\zeta\cdot {\rm rep@birth}\}$
         \item Step 2.3.3 Remove the column $\xi_{i,1}$ from $\matrixh$ 
         \item Step 2.3.4 Remove the column $\xi_{j,k-1}$ from $\matrixr$ if $d_{j,k-1}=\ell$ for some $j$
        \end{itemize}
    \end{itemize}
\end{itemize}

\section{Conclusion}
The cup product, the Massey products and the Steenrod operations are cohomology operations that gives the  cohomology vector
spaces the structure of a graded ring~\cite{MR1867354,mosher2008cohomology,massey1998higher,steenrod}.
Recently, Lupo et al.~\cite{lupo2018persistence} introduced invariants called Steenrod barcodes and devised algorithms for their computation, which were implemented in the software \texttt{steenroder}. 
 Our work complements the results in Lupo et al.~\cite{lupo2018persistence}, Contessoto et al.\cite{contessoto,cuparxivtwo} and M\'{e}moli et al.~\cite{memoli2}. We believe that the combined advantages of a fast algorithm and favorable stability properties make  cup modules and partition modules valuable additions to the topological data analysis pipeline.

We note that although the commonly used algorithms for computing ordinary persistence are worst case $O(n^3)$ time, in practice, on most datasets they run in nearly linear time~\cite{phat,bauer2022keeping,boissonnat2013compressed}. Likewise, although the theoretical complexity bound for the algorithms presented in this work is $O(dn^4)$ time, if implemented, we expect them to run in nearly cubic or even quadratic time on most datasets. 

Finally, it would be remiss not to mention the recent application of persistent cup products for quasi-periodicity detection in sliding window embedding of time-series data~\cite[Section 4.5]{luispolanco}. In fact, using extensive experimentation, in  \cite{luispolanco}, Polanco shows that cup product information leads to improved quasi-periodicity detection as compared to using only ordinary persistence as in~\cite{tralieperea}. However, in \cite{luispolanco} only the cup product of cocycle representatives of the two longest intervals in dimension $1$ barcode is used. A more recent work~\cite{gakharperea} uses the persistent K{\"u}nneth formula for quasiperiodicity detection. It is conceivable that cup modules could improve on the quasi-periodicity detection methods described in \cite{tralieperea,gakharperea,luispolanco}.

\bibliography{cupprod}

\begin{thebibliography}{10}

\bibitem{bauer2021ripser}
Ulrich Bauer.
\newblock Ripser: efficient computation of {V}ietoris--{R}ips persistence
  barcodes.
\newblock {\em Journal of Applied and Computational Topology}, 5(3):391--423,
  2021.

\bibitem{bauer2014clear}
Ulrich Bauer, Michael Kerber, and Jan Reininghaus.
\newblock Clear and compress: Computing persistent homology in chunks.
\newblock In {\em Topological methods in data analysis and visualization III},
  pages 103--117. Springer, 2014.

\bibitem{phat}
Ulrich Bauer, Michael Kerber, Jan Reininghaus, and Hubert Wagner.
\newblock Phat -- persistent homology algorithms toolbox.
\newblock {\em Journal of Symbolic Computation}, 78:76--90, 2017.

\bibitem{bauer2022keeping}
Ulrich Bauer, Talha~Bin Masood, Barbara Giunti, Guillaume Houry, Michael
  Kerber, and Abhishek Rathod.
\newblock Keeping it sparse: Computing persistent homology revisited.
\newblock {\em arXiv preprint arXiv:2211.09075}, 2022.

\bibitem{bauerschmahlsocg}
Ulrich Bauer and Maximilian Schmahl.
\newblock {Efficient Computation of Image Persistence}.
\newblock In Erin~W. Chambers and Joachim Gudmundsson, editors, {\em 39th
  International Symposium on Computational Geometry (SoCG 2023)}, volume 258 of
  {\em Leibniz International Proceedings in Informatics (LIPIcs)}, pages
  14:1--14:14, Dagstuhl, Germany, 2023. Schloss Dagstuhl -- Leibniz-Zentrum
  f{\"u}r Informatik.
\newblock \href {https://doi.org/10.4230/LIPIcs.SoCG.2023.14}
  {\path{doi:10.4230/LIPIcs.SoCG.2023.14}}.

\bibitem{belchi2021a}
Francisco Belch{\'\i} and Anastasios Stefanou.
\newblock A-infinity persistent homology estimates detailed topology from point
  cloud datasets.
\newblock {\em Discrete \& Computational Geometry}, pages 1--24, 2021.

\bibitem{boissonnat2013compressed}
Jean-Daniel Boissonnat, Tamal~K Dey, and Cl{\'e}ment Maria.
\newblock The compressed annotation matrix: An efficient data structure for
  computing persistent cohomology.
\newblock In {\em European Symposium on Algorithms}, pages 695--706. Springer,
  2013.

\bibitem{brooksbank2019testing}
Peter Brooksbank, E~O’Brien, and James Wilson.
\newblock Testing isomorphism of graded algebras.
\newblock {\em Transactions of the American Mathematical Society},
  372(11):8067--8090, 2019.

\bibitem{chazal2014persistence}
Fr{\'e}d{\'e}ric Chazal, Vin De~Silva, and Steve Oudot.
\newblock Persistence stability for geometric complexes.
\newblock {\em Geometriae Dedicata}, 173(1):193--214, 2014.

\bibitem{CEH07}
David Cohen-Steiner, Herbert Edelsbrunner, and John Harer.
\newblock Stability of persistence diagrams.
\newblock {\em Discrete Comput. Geom.}, 37(1):103--120, Jan 2007.
\newblock \href {https://doi.org/10.1007/s00454-006-1276-5}
  {\path{doi:10.1007/s00454-006-1276-5}}.

\bibitem{cohenimage}
David Cohen-Steiner, Herbert Edelsbrunner, John Harer, and Dmitriy Morozov.
\newblock Persistent homology for kernels, images, and cokernels.
\newblock In {\em Proceedings of the Twentieth Annual ACM-SIAM Symposium on
  Discrete Algorithms}, pages 1011--1020. SIAM, 2009.

\bibitem{contessoto}
Marco Contessoto, Facundo M{\'{e}}moli, Anastasios Stefanou, and Ling Zhou.
\newblock Persistent cup-length.
\newblock In {\em 38th International Symposium on Computational Geometry, SoCG
  2022, June 7-10, 2022, Berlin, Germany}, volume 224 of {\em LIPIcs}, pages
  31:1--31:17. Schloss Dagstuhl - Leibniz-Zentrum f{\"{u}}r Informatik, 2022.

\bibitem{cuparxivtwo}
Marco Contessoto, Facundo Mémoli, Anastasios Stefanou, and Ling Zhou.
\newblock Persistent cup-length, 2021.
\newblock URL: \url{https://arxiv.org/abs/2107.01553v3}, \href
  {https://doi.org/10.48550/ARXIV.2107.01553}
  {\path{doi:10.48550/ARXIV.2107.01553}}.

\bibitem{partsbound}
Wladimir de~Azevedo~Pribitkin.
\newblock Simple upper bounds for partition functions.
\newblock {\em The Ramanujan Journal}, 18(1):113--119, 2009.

\bibitem{de2011dualities}
Vin De~Silva, Dmitriy Morozov, and Mikael Vejdemo-Johansson.
\newblock Dualities in persistent (co) homology.
\newblock {\em Inverse Problems}, 27(12):124003, 2011.

\bibitem{DW22}
Tamal~K. Dey and Yusu Wang.
\newblock {\em Computational Topology for Data Analysis}.
\newblock Cambridge University Press, 2022.

\bibitem{EH10}
Herbert Edelsbrunner and John Harer.
\newblock {\em Computational Topology: An Introduction}.
\newblock Applied Mathematics. American Mathematical Society, 2010.

\bibitem{partitioncode}
David Eppstein.
\newblock Python code for generating partitions.
\newblock \url{https://code.activestate.com/recipes/218332/}.
\newblock Accessed: 2023-11-30.

\bibitem{gakharperea}
Hitesh Gakhar and Jose~A. Perea.
\newblock Sliding window persistence of quasiperiodic functions.
\newblock {\em Journal of Applied and Computational Topology}, 8(1):55--92,
  2024.

\bibitem{MR1867354}
Allen Hatcher.
\newblock {\em Algebraic topology}.
\newblock Cambridge University Press, Cambridge, 2002.

\bibitem{hausmann2014mod}
Jean-Claude Hausmann.
\newblock {\em Mod two homology and cohomology}, volume~10.
\newblock Springer, 2014.

\bibitem{herscovich2018higher}
Estanislao Herscovich.
\newblock A higher homotopic extension of persistent (co)homology.
\newblock {\em Journal of Homotopy and Related Structures}, 13(3):599--633,
  2018.

\bibitem{lupo2018persistence}
Umberto Lupo, Anibal~M. Medina-Mardones, and Guillaume Tauzin.
\newblock Persistence {S}teenrod modules.
\newblock {\em Journal of Applied and Computational Topology}, 6(4):475--502,
  2022.

\bibitem{gudhi}
Cl{\'e}ment Maria, Jean-Daniel Boissonnat, Marc Glisse, and Mariette Yvinec.
\newblock The {G}udhi library: Simplicial complexes and persistent homology.
\newblock In Hoon Hong and Chee Yap, editors, {\em Mathematical Software --
  ICMS 2014}, pages 167--174, Berlin, Heidelberg, 2014. Springer Berlin
  Heidelberg.

\bibitem{massey1998higher}
William~S. Massey.
\newblock Higher order linking numbers.
\newblock {\em Journal of Knot Theory and Its Ramifications}, 7:393--414, 1998.

\bibitem{memoli2}
Facundo M{\'e}moli, Anastasios Stefanou, and Ling Zhou.
\newblock Persistent cup product structures and related invariants.
\newblock {\em Journal of Applied and Computational Topology}, 2023.

\bibitem{mosher2008cohomology}
Robert~E. Mosher and Martin~C. Tangora.
\newblock {\em Cohomology operations and applications in homotopy theory}.
\newblock Courier Corporation, 2008.

\bibitem{luispolanco}
Luis Polanco.
\newblock {\em Applications of persistent cohomology to dimensionality
  reduction and classification problems}.
\newblock Phd thesis, Michigan State University, 2022.
\newblock URL: \url{https://doi.org/doi:10.25335/exk0-fs44}.

\bibitem{steenrod}
Norman~E Steenrod.
\newblock Products of cocycles and extensions of mappings.
\newblock {\em Annals of Mathematics}, pages 290--320, 1947.

\bibitem{tralieperea}
Christopher~J. Tralie and Jose~A. Perea.
\newblock (quasi)periodicity quantification in video data, using topology.
\newblock {\em SIAM Journal on Imaging Sciences}, 11(2):1049--1077, 2018.

\bibitem{yarmola2010persistence}
Andrew Yarmola.
\newblock {\em Persistence and computation of the cup product}.
\newblock Undergraduate honors thesis, Stanford University, 2010.

\bibitem{zomorodian2004computing}
Afra Zomorodian and Gunnar Carlsson.
\newblock Computing persistent homology.
\newblock In {\em Proceedings of the twentieth annual symposium on
  Computational geometry}, pages 347--356, 2004.

\end{thebibliography}

\appendix

\newpage

\section{Computing persistent cup-length} \label{sec:better}
This section expands  Section~\ref{sec:originalbetter}.
The \emph{cup length} of a ring is defined as the maximum number of multiplicands that together give a nonzero product in the ring.
Let $\Int_{\ast} $ denote the set of all closed  intervals of $\mathbb{R}$, and let $\Int_{\circ}$ denote the set of all the open-closed intervals of $\mathbb{R}$ of the form $(a,b]$.
Let $\mathcal{F}$ be an $\mathbb{R}$-indexed filtration of simplicial complexes. 
The \emph{persistent cup-length function} $\cupprod_{\bullet}:\Int_{\ast}\to\NBB$  (introduced in \cite{contessoto,cuparxivtwo})
is defined as the function 
from the set of closed intervals to the set of non-negative integers.\footnote{For simplicity and without loss of generality, we define persistent cup-length only for intervals in $\Int_{\ast}$, and persistent cup-length diagram only for intervals in $\Int_{\circ}$.} Specifically, it assigns to each interval $[ a,b]$, the cup-length of the image ring  $\im\big(\Homol^*(\complex)[ a,b]\big)$, which is the ring $\im\big(\Homol^*(\complex_b)\to\Homol^*(\complex_a)\big)$. 

Let the restriction of a cocycle $\xi$ to index $k$ be $\xi^{k}$.
We say that a cocycle $\zeta$ is defined at $p$ if there exists a cocycle $\xi$ in $\complex_{q}$ for $q\geq p$ and $\zeta = \xi^{p}$.

For a persistent cohomology basis $\PCB$, we say that $[d,b)$ is  \emph{a supported interval of length $k$ for $\PCB$} if there exists cocycles $\xi_1,\dots,\xi_k \in \PCB$ such that the  product cocycle $\eta^s = \xi_1^s \smile \dots \smile \xi_k^s$ is nontrivial for every $s \in [d,b)$ and $\eta^s$ either does not exist or is trivial outside of $[d,b)$. In this case, we say that $[d,b)$ \emph{is supported by} $\{\xi_1,\dots,\xi_k\}$.
The max-length of a supported interval $[d,b)$, denoted by $\ell_{\Omega}([d,b))$, is defined as 
\[\ell_{\Omega}([d,b)) = \max \{k \in \mathbb{N} \mid \exists \xi_1,\dots,\xi_k \in \PCB  \text{ such that } (d,b] \text{ is supported by } \{\xi_1,\dots,\xi_k\} \}. \]

Let $\Int_{\PCB}$ be the set of supported intervals of $\PCB$.
In order to compute the persistent cup-length function,
Contessoto et al.~\cite{contessoto}  define a notion called the \emph{persistent cup-length diagram}, which is a function  $\mathbf{dgm}^{\smile}_{\PCB}: \Int_{\circ} \to \NBB$,  that assigns to every interval $[d,b)$ in $\Int_{\Omega} \subset \Int_{\circ}$ its max-length $\ell_{\Omega}([d,b))$, and assigns zero to every interval in $\Int_{\circ} \setminus \Int_{\Omega} $. 

It is worth noting that unlike the order-$k$ product persistence modules, the persistent cup-length diagram is not a topological invariant as it depends on the choice of representative cocycles. While the persistent cup-length diagram is not useful on its own, in Contessoto et al.~\cite{contessoto}, it serves as an intermediate in computing the persistent cup-length (a stable topological invariant) due to the  following theorem.

\begin{theorem}[Contessoto et al.~\cite{contessoto}]
\label{thm:length_diagram}
Let $\mathcal{F}$ be a filtered simplicial complex, and let $\PCB$ be a persistent cohomology basis for $\mathcal{F}$.
The persistent cup-length function $\cupprod_{\bullet}$ can be retrieved from the persistent cup-length diagram $\mathbf{dgm}_{\PCB}^{\smile}$ for any $( a,b]\in \Int_{\circ}$ as follows.
\begin{equation}\label{eq:tropical_mobuis}
    \cupprod_{\bullet}([ a,b])=  \max_{(c,d]\supset [ a,b]}\mathbf{dgm}_{\PCB}^{\smile}((c,d]). 
\end{equation}
\end{theorem}

Given a $P$-indexed filtration $\mathcal{F}$, let $V^k_\bullet$ denote its
persistent $k$-cup module. The following result appears as Proposition 5.9 in \cite{cuparxivtwo}. We provide a short proof in our notation for the sake of completeness.

\begin{proposition}[Contessoto et al.~\cite{cuparxivtwo}] \label{prop:easyprop}
 $\cupprod_{\bullet}([a,b])={\rm argmax}\{k \mid \rk_{V^k_\bullet}([a,b])\not = 0\}$.   
\end{proposition}
\begin{proof}
    $\cupprod_{\bullet}([a,b])= k \quad\Longleftrightarrow \quad$ 
 \begin{inparaenum}
\item     There exists a set of cocycles $\{\xi_1,\dots,\xi_k\}$ that are defined at $b$ and 
    $\xi_1^s \smile \dots \smile 
 \xi_k^s$ is nontrivial for all $s\in [a,b]$ \item For any set of $k+1$ cocycles $\{\zeta_1,\dots,\zeta_{k+1}\}$ that are defined at $b$, the product $\zeta_1^s \smile \dots \smile 
 \zeta_{k+1}^s$ is zero for some $s \in [a,b]$. 
 \end{inparaenum}  $\quad\Longleftrightarrow\quad$ 
$\rk_{V^k_\bullet}([a,b])\not = 0$ and  $\rk_{V^{k+1}_\bullet}([a,b])  = 0$.
\end{proof}

Given a filtered complex $\complex_{\bullet}: \complex_{1} \hookrightarrow \complex_{2} \hookrightarrow \dots $,
Contessoto et al.~\cite{contessoto} define its $p$-truncation as the filtration $\complex_{\bullet}^p: \complex_{1}^p \hookrightarrow \complex_{2}^p \hookrightarrow \dots $, where for all $i$, $\complex_{i}^p$ denotes the $p$-skeleton of $\complex_{i}$. We now compare the complexities of computing the persistent cup-length using the algorithm described in Contessoto et al.~\cite{contessoto} against computing it with our approach. 

Assume that $\complex$ is a $d$-dimensional complex of size $n$, and let $n_p$ denote the number of simplices in the $p$-skeleton of $\complex$.
Let $\mathcal{F}$ be a filtration of $\complex$ and let $\mathcal{F}_p$ be the $p$-truncation of $\mathcal{F}$.
Then, according to Theorem~20 in Contessoto et al.~\cite{contessoto}, 
    using the persistent cup-length diagram,  
    \begin{inparaenum}
            \item the persistent cup-length   of $\mathcal{F}$ can be computed in $O(n^{d+2})$ time,
    \item the persistent cup-length  of $\mathcal{F}_p$ can be computed in $O(n_{p}^{p+2})$ time. 
    \end{inparaenum}

In contrast, as noted in \Cref{sec:alg},
the barcodes of all the persistent $k$-cup modules for  $k \in \{2,\dots,p\}$ can be computed in $O(p \, n^4)$ time.  Note that $\rk_{V^k_\bullet}([a,b])\not = 0$ if and only if there exists an interval $(x,y]$ in $B(V^k_\bullet)$ such that $(x,y] \supset [a,b]$.  This suggests a simple algorithm to compute $\cupprod_{\bullet}$ from the barcodes of persistent $k$-cup modules for $k\in \{ 2, \dots, n \}$, that is, one finds the largest $k$ for which there exists an interval $(x,y] \in B(V^k_\bullet)$ such that $(x,y] \supset [a,b]$. Since the size of $B(V^k_\bullet)$, for every $k\in[n]$, is $O(n)$, the algorithm for extracting the persistent cup-length from the barcode of persistent $k$-cup modules for $k\in \{2,\dots, d\}$ runs in $O(n^2)$ time. 
Thus, using the algorithms described in \Cref{sec:alg-kproduct}, the persistent cup-length of a ($p$-truncated) filtration can be computed in $O(d n^4)$ ($O(p n^4)$) time, which is strictly better than the coarse bound for the algorithm in \cite{contessoto} for $d\geq 3$. 


\section{Correctness of {\sc OrderkCupPers} } \label{sec:orderk}

In this section, we provide a brief sketch of correctness of {\sc OrderkCupPers}. The statements of lemmas and their proofs are analogous to the case when $k=2$ treated in the main body of the paper. 

\begin{remark}\label{rem:defmaptwo}
Since $\productmodulekk$ is a submodule of $\ordmodule$, the
 structure maps  of $\productmodulekk$ for every $i \in I$, namely, $\im \Homol^{*}(\smile_{i}^k) \to \im \Homol^{*}(\smile_{i-1}^k)$  are given by restrictions of $\cohommap_i$ to $\im \Homol^{*}(\smile_{i}^k)$.
\end{remark}

\begin{definition}
    For any $i\in \{0,\dots,n\}$, a nontrivial cocycle  $\zeta \in \mathsf{Z}^{\ast}(\complex_i)$  is said to be an \emph{order-$k$ product cocycle of $\complex_i$} if $[\zeta] \in\im \Homol^{*}(\smile_{i}^k)$.
\end{definition}

\begin{proposition} \label{prop:mainstructk}
    For a filtration $\mathcal{F}$ of simplicial complex $\complex$, the birth points of $B(\productmodulekk)$ are a subset of the birth points of $B(\ordmodule)$, and the death points of $B(\productmodulekk)$ are a subset of the death points of $B(\ordmodule)$. 
\end{proposition}
\begin{proof}
   Let $\{(d_{i_j},b_{i_j}] \mid j \in [k] \}$ be (not necessarily distinct) intervals in $B(\ordmodule)$, where $b_{i_{j+1}} \geq b_{i_{j}}$ for $j \in [k-1]$. Let $\xi_{i_j}$ be  a representative for $(d_{i_j},b_{i_j}]$ for $j \in [k]$.

    If $\xi_{i_1} \smile \xi_{i_2}^{b_{i_1}} \smile  \dots\smile\xi_{i_k}^{b_{i_1}}$ is trivial, then by the functoriality of cup product,
\begin{align*}
     \cochainmap_{b_{i_1},r}(\xi_{i_1} \smile \xi_{i_2}^{b_{i_1}} \smile  \dots\smile \xi_{i_k}^{b_{i_1}}) &= \cochainmap_{b_{i_1},r}(\xi_{i_1}) \smile \cochainmap_{b_{i_1},r}(\xi_{i_2}^{b_{i_1}}) \smile \dots \smile \cochainmap_{b_{i_1},r}(\xi_{i_k}^{b_{i_1}}) \\
     &= \xi_{i_1}^{r} \smile \xi_{i_2}^{r} \smile \dots \smile \xi_{i_k}^{r} 
\end{align*}

    is trivial  $\forall r < b_{i_1}$. Writing contrapositively, if  $\exists r < b_{i_1}$ for which $\xi_{i_1}^{r} \smile \xi_{i_2}^{r} \smile \dots \smile   \xi_{i_k}^{r}$ is nontrivial, then $\xi_{i_1} \smile \xi_{i_2}^{b_{i_1}} \smile  \dots\smile\xi_j^{b_{i_1}}$ is nontrivial.
    Noting that $\im \Homol^{*}(\smile_{\ell}^k)$ for any $\ell \in \{0,\dots,n\}$ is generated by $\{[\xi_{i_1}^\ell] \smile \{[\xi_{i_2}^\ell] \smile \dots \smile [\xi_{i_k}^
    \ell] \mid \xi_{i_j} 
\in \Omega_\complex \text{ for }j\in[k]\}$, it follows that $b$ is the birth point of an interval in $B(\productmodulekk)$ only if it is  the birth point of an interval in $B(\ordmodule)$, proving the first claim. 

Let $\Omega'_{j+1} = \{ [\tau_1],\dots, [\tau_\ell]\}$ be a basis for $\im \Homol^{*}(\smile_{j+1}^k)$. Then, $\Omega'_{j+1}$ extends to a basis $\Omega_{j+1}$ of $\Homol^{\ast}(\complex_{j+1})$.
If $j$ is not a death index in $B(\ordmodule)$, then $\cochainmap_{j+1}(\tau_1), \dots, \cochainmap_{j+1}(\tau_\ell)$ are all nontrivial and linearly independent. Using \Cref{rem:defmap}, it follows that $j$ is not a death index in $B(\productmodulekk)$, proving the second claim. 
 \end{proof}

\begin{corollary} \label{cor:onediek}
If $d$ is a death index in $B(\productmodulekk)$, then at most one bar of $B(\productmodulekk)$ has death index $d$.
\end{corollary}
\begin{proof}
The proof is identical to \Cref{cor:onedie}.
\end{proof}

Let $C_{b}$ be the vector space of order-$k$ product cocycle classes created at index $b$.
We note that for a birth index $b \in \{0,\dots,n\}$,
$C_b$ is a subspace of $\Homol^{\ast}(\complex_b)$ which can be written as

\begin{equation} \label{eqn:collectkcocycles}
C_b = 
\begin{cases} \langle [\xi_{i_1}] \smile \dots \smile [\xi_{i_k}] \mid \xi_{i_j} \text{ for }j\in [k] \text{ are essential cocycles of } \Homol^\ast(\complex_\bullet)\rangle & \text{ if } b=n \\ 
 \langle [\xi_{i_1}] \smile \dots \smile  \ [\xi_{i_k}^b] \mid \xi_{i_1}\text{ is born at $b$}\,\, \& \,\,\xi_{i_j}\text{ for } j\neq 1 \text{ is  born at  index} \geq b \rangle & \text{ if } b<n
\end{cases} 
\end{equation}

For a birth index $b$, let $\matrixc_{b}$ be the submatrix of $\matrixb$ formed by representatives whose classes generate $C_b$, augmented to $\matrixb$ in Steps 2.1.2 (i) and 2.2.2 (i)  when $k=b$ in the outer \textbf{for} loop of Step 2.

\begin{theorem}
Algorithm {\sc OrderkCupPers} correctly computes the barcode of persistent $k$-cup modules.
\end{theorem}
\begin{proof} The proof is nearly identical to \Cref{thm:maincorrect}. The key difference (from \Cref{thm:maincorrect}) is in how the submatrix $\matrixc_{b}$ of $\matrixb$ that stores the linearly independent order-$k$ product cocycles born at $\ell=b$ in Steps 2.1 and 2.2 is built. It is easy to check that the classes of the cocycle vectors in $\matrixc_{b}$ augmented to $\matrixb$ in Steps 2.1 and 2.2 generate the space $C_b$ described in \Cref{eqn:collectkcocycles}.
\end{proof}

\end{document}